\PassOptionsToPackage{unicode}{hyperref}
\PassOptionsToPackage{hyphens}{url}
\PassOptionsToPackage{dvipsnames,svgnames,x11names}{xcolor}

\newcommand{\changed}[1]{{\color{black}#1}}
\newcommand{\lp}{\ell'}
\documentclass[12pt]{article}

\usepackage{amsmath,amsfonts,amssymb,amsthm}
\usepackage{paralist}

\usepackage{algorithm}
\usepackage{algorithmic}
\usepackage{bm}

\usepackage{iftex}
\ifPDFTeX
  \usepackage[T1]{fontenc}
  \usepackage[utf8]{inputenc}
  \usepackage{textcomp} 
\else 
  \defaultfontfeatures{Scale=MatchLowercase}
  \defaultfontfeatures[\rmfamily]{Ligatures=TeX,Scale=1}
\fi
\usepackage{lmodern}
\ifPDFTeX\else  
\fi
\IfFileExists{upquote.sty}{\usepackage{upquote}}{}
\IfFileExists{microtype.sty}{
  \usepackage[]{microtype}
  \UseMicrotypeSet[protrusion]{basicmath} 
}{}
\makeatletter
\@ifundefined{KOMAClassName}{
  \IfFileExists{parskip.sty}{%
    \usepackage{parskip}
  }{
    \setlength{\parindent}{0pt}
    \setlength{\parskip}{6pt plus 2pt minus 1pt}}
}{
  \KOMAoptions{parskip=half}}
\makeatother
\usepackage{xcolor}
\makeatletter
\ifx\paragraph\undefined\else
  \let\oldparagraph\paragraph
  \renewcommand{\paragraph}{
    \@ifstar
      \xxxParagraphStar
      \xxxParagraphNoStar
  }
  \newcommand{\xxxParagraphStar}[1]{\oldparagraph*{#1}\mbox{}}
  \newcommand{\xxxParagraphNoStar}[1]{\oldparagraph{#1}\mbox{}}
\fi
\ifx\subparagraph\undefined\else
  \let\oldsubparagraph\subparagraph
  \renewcommand{\subparagraph}{
    \@ifstar
      \xxxSubParagraphStar
      \xxxSubParagraphNoStar
  }
  \newcommand{\xxxSubParagraphStar}[1]{\oldsubparagraph*{#1}\mbox{}}
  \newcommand{\xxxSubParagraphNoStar}[1]{\oldsubparagraph{#1}\mbox{}}
\fi
\makeatother

\usepackage{longtable,booktabs,array}
\usepackage{calc} 
\usepackage{etoolbox}
\makeatletter
\patchcmd\longtable{\par}{\if@noskipsec\mbox{}\fi\par}{}{}
\makeatother
\IfFileExists{footnotehyper.sty}{\usepackage{footnotehyper}}{\usepackage{footnote}}
\makesavenoteenv{longtable}
\usepackage{graphicx}
\makeatletter
\def\maxwidth{\ifdim\Gin@nat@width>\linewidth\linewidth\else\Gin@nat@width\fi}
\def\maxheight{\ifdim\Gin@nat@height>\textheight\textheight\else\Gin@nat@height\fi}
\makeatother
\setkeys{Gin}{width=\maxwidth,height=\maxheight,keepaspectratio}
\makeatletter
\def\fps@figure{htbp}
\makeatother

\makeatletter
\@ifpackageloaded{caption}{}{\usepackage{caption}}
\AtBeginDocument{%
\ifdefined\contentsname
  \renewcommand*\contentsname{Table of contents}
\else
  \newcommand\contentsname{Table of contents}
\fi
\ifdefined\listfigurename
  \renewcommand*\listfigurename{List of Figures}
\else
  \newcommand\listfigurename{List of Figures}
\fi
\ifdefined\listtablename
  \renewcommand*\listtablename{List of Tables}
\else
  \newcommand\listtablename{List of Tables}
\fi
\ifdefined\figurename
  \renewcommand*\figurename{Figure}
\else
  \newcommand\figurename{Figure}
\fi
\ifdefined\tablename
  \renewcommand*\tablename{Table}
\else
  \newcommand\tablename{Table}
\fi
}
\@ifpackageloaded{float}{}{\usepackage{float}}
\floatstyle{ruled}
\@ifundefined{c@chapter}{\newfloat{codelisting}{h}{lop}}{\newfloat{codelisting}{h}{lop}[chapter]}
\floatname{codelisting}{Listing}

\makeatother
\makeatletter
\@ifpackageloaded{caption}{}{\usepackage{caption}}
\@ifpackageloaded{subcaption}{}{\usepackage{subcaption}}
\makeatother

\ifLuaTeX
  \usepackage{selnolig}  
\fi
\usepackage[]{natbib}
\usepackage{bookmark}

\IfFileExists{xurl.sty}{\usepackage{xurl}}{} 
\hypersetup{
  pdftitle={Quasi-Monte Carlo with one categorical variable},
  pdfauthor={Author 1; Author 2},
  pdfkeywords={3 to 6 keywords, that do not appear in the title},
  colorlinks=true,
  linkcolor={blue},
  filecolor={Maroon},
  citecolor={Blue},
  urlcolor={Blue},
  pdfcreator={LaTeX via pandoc}}

\date{January 2026}

\newcommand{\anon}{1}

\allowdisplaybreaks

\newcommand{\bsa}{\boldsymbol{a}}
\newcommand{\bsb}{\boldsymbol{b}}
\newcommand{\bsc}{\boldsymbol{c}}
\newcommand{\bsu}{\boldsymbol{u}}
\newcommand{\bsx}{\boldsymbol{x}}

\newcommand{\bsz}{\boldsymbol{z}}

\newcommand{\natu}{\mathbb{N}}
\newcommand{\real}{\mathbb{R}}

\newcommand{\rd}{\mathrm{\, d}}

\newcommand{\var}{{\mathrm{Var}}}

\newcommand{\vol}{{\mathbf{vol}}}

\DeclareMathOperator*{\argmax}{argmax}
\DeclareMathOperator*{\argmin}{argmin}

\renewcommand{\ge}{\geqslant}
\renewcommand{\le}{\leqslant}

\newcommand{\dnorm}{\mathcal{N}}

\newcommand{\dunif}{\mathbf{U}} 
\newcommand{\dustd}{\mathbf{U}} 

\newcommand{\e}{{\mathbb{E}}} 

\newcommand{\hk}{{\mathrm{HK}}}

\renewcommand{\vol}{\mathrm{Vol}}

\newcommand{\bsn}{\boldsymbol{n}}

\newcommand{\bstau}{\boldsymbol{\tau}}
\newcommand{\bszero}{\boldsymbol{0}}

\renewcommand{\le}{\leqslant}
\renewcommand{\ge}{\geqslant}
\renewcommand{\leq}{\leqslant}
\renewcommand{\geq}{\geqslant}

\newcommand{\mc}{\mathrm{MC}}
\newcommand{\rqmc}{\mathrm{RQMC}}
\newcommand{\adj}{\mathrm{ADJ}}

\newcommand{\ind}{\boldsymbol{1}} 

\newcommand{\INPUT}{\item[\textbf{Input:}]}
\newcommand{\OUTPUT}{\item[\textbf{Output:}]}

\newcommand{\giv}{\!\mid\!}
\newcommand{\bsell}{\boldsymbol{\ell}}

\newcommand{\bsk}{\boldsymbol{k}}

\newcommand{\simiid}{\stackrel{\mathrm{iid}}{\sim}}

\newtheorem{proposition}{Proposition}

\begin{document}

\def\spacingset#1{\renewcommand{\baselinestretch}%
{#1}\small\normalsize} \spacingset{1}


\if1\anon
{
  \title{\bf Quasi-Monte Carlo with one categorical variable}
  \author{Valerie N. P. Ho\\
  Stanford  University   \\and\\
    Art B. Owen\thanks{
      This work was supported by the National Science Foundation
      under grant DMS-2152780}\\  
 Stanford University\\and\\
 Zexin Pan\\
 Zhejiang University
}
  \maketitle
} \fi

\if0\anon
{
  \bigskip
  \bigskip
  \bigskip
  \begin{center}
    {\LARGE\bf Title}
\end{center}
  \medskip
} \fi

\begin{abstract}
We study randomized quasi-Monte Carlo (RQMC)
estimation of a multivariate integral where one of the variables
takes only a finite number of values. This problem arises 
when the variable of integration is drawn from a mixture
distribution as is common in importance sampling and also
arises in some recent work on transport
maps. We find that when integration error decreases
at an RQMC rate that it is then beneficial to oversample
the smallest mixture components instead of using
a proportional allocation. The optimal allocations
depend on the possibly unknown convergence rates.
Designing the sample with an incorrect assumption
on the rate still attains that convergence rate, with an inferior
implied constant.  The penalty for using a pessimistic rate
is typically higher than for using an optimistic one.
We also find that for the most
accurate RQMC sampling methods, it is advantageous
to arrange that our $n=2^m$ randomized Sobol' points 
split into subsample sizes that are also powers of~$2$.
\end{abstract}

\noindent%


\section{Introduction}

We study randomized quasi-Monte Carlo (RQMC) integration in a setting where
one of the input variables is categorical, 
taking only some finite number $L\ge2$ of different values.
The main motivation is mixture sampling where the variables of interest come from $L$ different distributions representing
$L$ different scenarios. That is, $\bsx\sim P_\ell$ with probability
$\alpha_\ell>0$ where $\sum_{\ell=1}^L\alpha_\ell=1$.
We want to estimate $\mu = \e(g(\bsx))=\sum_{\ell=1}^L\alpha_\ell \e(g(\bsx)\giv \bsx\sim P_\ell)$. 

A straightforward approach to this problem is to simply replace a
standard Monte Carlo (MC) workflow with RQMC sampling.
In that workflow, one random variable selects \changed{component} $\ell$ with
probability $\alpha_\ell$.
\changed{Then $s$ more} random variables generate $\bsx\sim P_\ell$. 
\changed{RQMC estimates over $s+1$ dimensions converge at a better rate than MC estimates do. Then the optimal sampling allocations become more nearly equal than proportional sampling, and we modify RQMC to take advantage of this phenomenon.  We see reduced variance with those choices, and under regularity conditions, non-proportional  sampling of the $P_\ell$ can even improve the convergence rate.  We also show a minimaxity property for equal or nearly equal sampling fractions.}

\changed{The input points leading to $\bsx\sim P_\ell$ come from $L$ disjoint subsets or strata in $[0,1]^{s+1}$. From here on we refer to the mixture components as strata to avoid confusion with vector components.}

An alternative to the standard MC workflow 
is to compute $L$ different independently sampled RQMC estimates,
one per stratum, and 
combine them according to their mixture probabilities $\alpha_\ell$.
That approach
misses a potential variance reduction from RQMC sampling 
when the within-stratum estimates are negatively correlated.
Such negative correlation need not happen; we give a reason to expect \changed{them} at the end of the paper.

RQMC does best under favorable regularity assumptions, such
as continuity, on the integrand. 
The presence of the categorical variable makes the resulting
integrand discontinuous.  That discontinuity is however 
axis-parallel
and so it is a QMC-friendly discontinuity in the terminology of
\cite{wang:tan:2013}.  In some very regular cases,
RQMC can attain a mean squared error (MSE) of $O(n^{-3+\epsilon})$.
Such a rate can only be attained over a sequence of special sample
sizes that grow at least exponentially \citep{sobo:1998,quadconstraint-print}. This rate can be attained by scrambled Sobol' points with sample size 
equal to a power of $2$.  
Even an integrand with a QMC-friendly axis-parallel discrepancy is not 
generally integrated with an error of this rate.  By using an importance
sampling \cite[Chapter~9]{mcbook} adjustment, we can arrange to 
obtain $n_\ell = 2^{m_\ell}$ points 
from $P_\ell$ while also having $n = \sum_{\ell=1}^Ln_\ell$ 
be a power of $2$. This can then give the $O(n_\ell^{-3+\epsilon})$ rate within every stratum and hence overall.

This paper is organized as follows.
Section \ref{sec:notation} introduces our notation, gives some
background on RQMC, and proves some results about within-stratum
discrepancies.
Section \ref{sec:inefficiencies} describes how to allocate
samples to strata. We see that in RQMC sampling, it is advantageous to use
more nearly equal allocations than those given by the
stratum sizes.
Section \ref{sec:powtwo} shows how to partition a sample
size $n=2^m$ into $L$ within-stratum sample sizes
$n_\ell =2^{m_\ell}$.
Section \ref{sec:minimax} establishes some minimax
variance results from using equal or nearly equal sample
sizes for the strata.
Section \ref{sec:toy} shows a simple toy example with $L=8$
strata where RQMC with $n_\ell$ taken to be powers of two
does best.
Section \ref{sec:flood} uses a mixture model version of the Saint-Venant
flood example from \cite{limb:dero:2010} where non-nominal sampling
proportions do best. \changed{The code used to produce the results for both examples is publicly available at \url{https://github.com/hoval58/qmc4categorical}. }
Section \ref{sec:discussion} has some discussion and conclusions.
\changed{There is an appendix for the proofs.}

We finish this introduction by describing some related
work in the literature.  
The key differences between
our work and the prior work we found \changed{are} in how the
problem has been formulated, and how the values are
sampled. 

We use the first component of our RQMC points to select
the stratum to sample from.  
\cite{keller2012parallel} used that approach in parallel computation.
Their first component was
used to determine which processor $\ell\in\{1,2,\dots,L\}$ 
got the $i$'th job.  The remaining components of their
QMC points sampled that job on processor $\ell$.
In that setting, $\alpha_\ell=1/L$ and
it makes sense to choose $L$ to be a power of $2$ at least
for Sobol' sequences.

\cite{cui:dick:pill:2025} consider a target distribution
on a bounded hyperrectangle $[\bsa,\bsb]\subset\real^d$.
The target probability density function $\pi$ is bounded.
They approximate it by a sum of products of hat functions.
That is equivalent to a mixture of distributions whose
components have independent triangle distributions.
They consider adaptive methods to select those triangle
densities given an unnormalized $\bar\pi\propto\pi$. 
For a possibly large number of mixture
components, they sample each component with $n_j\doteq n\alpha_j$
input values.
They use the first $n_j$ points of the same QMC point
set for each of the strata.  

We use different RQMC points
in each stratum.  Our inputs are randomized inputs, 
and so we expect better error cancellation than
we would get using the same inputs for every stratum.
Since our analysis is based on variances, we do
not require a bounded domain or integrand.
\changed{As mentioned above, we use non-proportional allocation.}

\cite{kleb:sull:2026} describe using transport maps
from a $\dnorm(0,I)$ reference distribution 
to sample a mixture like $P=\sum_{\ell=1}^L\alpha_j\dnorm(\mu_j,\Sigma_j)$.
They construct a continuous
mapping over $t\in[0,1]$ with $\dnorm(0,I)$ at $t=0$
reaching $P$ at $t=1$ and solve an ordinary
differential equation along that path.  For $\ell<L$,
they take the first $n_\ell = \lfloor n\alpha_\ell\rfloor$
points of a quasi-Monte Carlo sequence and use those
to sample from component $\ell$.
Then $n_L=n-\sum_{\ell=1}^{L-1}n_\ell$.


\section{Background}\label{sec:notation}

Here we introduce notation to describe the different
variables used in the integration workflow. Then
we define digital nets and discuss their randomizations.
For a broader survey of RQMC see \cite{lecu:lemi:2002}.

We use $\natu$ for the set of positive integers.
We use $\ind\{A\}$ for the indicator function that
equals $1$ when $A$ holds and is $0$ otherwise.

\subsection{Integration variables}
Our goal is to estimate \changed{this weighted sum of $d$-dimensional integrals}
$$
\mu = \sum_{\ell=1}^L\alpha_\ell \int_{\real^d}g(\bsx)p_\ell(\bsx)\rd\bsx
$$
where $p_\ell$ is the probability density function of $P_\ell$
and $g(\cdot)$ is the integrand of interest.
We assume that we have functions $\phi_\ell$ with
$\phi_\ell(\bsu)\sim P_\ell$ for $\bsu\sim\dunif(0,1)^s$.
Many such functions can be obtained using methods like those in
\cite{devr:1986}.
The dimension $s$ need not equal $d$. 

We will use RQMC points $\bsz\in[0,1]^{s+1}$.
We partition $\bsz$ into its first component $v$
that we use to select a stratum $\ell$
and the subsequent $s$ components $\bsu$ that
we use to sample $\bsx\sim P_\ell$.
That is, $\bsz = (v,u_1,u_2,\dots,u_s)$\changed{, and our $n$ sample points are $\bsz_i=(v_i,u_{i1},\dots,u_{is})$.}
For a  function $\bsell:[0,1]\to\{1,2,\dots,L\}$, we
will sample the $i$'th point using sample stratum $\bsell(v_i)$.  The pre-image of 
each $\ell\in\{1,2,\dots,L\}$ under $\bsell$ will always be a sub-interval of $[0,1]$.
We let $\beta_\ell$ be the length of that subinterval.
We abbreviate $\bsell(v_i)$ to $\ell(i)$ in some expressions.
Having selected stratum $\ell(i)$, we draw
$\bsx_i\in\real^d$ via $\bsx_i=\phi_{\ell(i)}(\bsu_i)$.

We now estimate $\mu =\e(g(\bsx))$ by
\begin{align}\label{eq:muhatg}
\hat\mu = 
\frac1n\sum_{i=1}^n 
\omega(v_i)
g(\bsx_i)\quad\text{where}\quad
\omega(v_i)=\frac{\alpha_{\ell(i)}}{\beta_{\ell(i)}},
\end{align}
is an importance sampling weight.  This weight only depends on the stratum
$\ell$ and so we use $\omega_\ell$ for $\omega(v)$ when $\bsell(v)=\ell$.
Two other ways to write this estimate will be useful.
Because both $v_i$ and
$\bsx_i$ are functions of $\bsz_i$ we write
\begin{align}\label{eq:muhatf}
\hat\mu = \frac1n\sum_{i=1}^n f(\bsz_i), 
\end{align}
for $$f(\bsz_i) = \omega(v_i)g(\bsx_i)
=\omega(v_i)g(\phi_{\bsell(v_i)}(\bsu_i)).$$
This $f$ is our $s+1$ dimensional integrand for which we will use RQMC.
It satisfies
$$
\mu = \int_{[0,1]^{s+1}} f(\bsz)\rd\bsz.
$$

\changed{Because each RQMC point $\bsz_i$ is uniformly distributed on $[0,1]^{s+1}$, the estimator $\hat{\mu}$ is unbiased for $\mu$.} 

We also define the within-stratum integrands individually
as functions of $[0,1]^s$.
The within-stratum integrands and integrals are
$$
h_\ell(\bsu) =  g(\phi_\ell(\bsu))
\quad\text{and}\quad
\mu_\ell = \int_{[0,1]^s}h_\ell(\bsu)\rd\bsu$$
respectively. We let 
\begin{align}\label{eq:defSell}
S_\ell =\{1\le i\le n\mid \bsell(v_i)=\ell\}
\end{align}
index the set of sample points for stratum $\ell$ and define
$n_\ell$ to be its cardinality.
Now we can write
\begin{align*}
\hat\mu 
&= \frac1n\sum_{i=1}^n\omega(v_i)g(\bsx_i)
= \frac1n\sum_{i=1}^n\sum_{\ell=1}^L\ind\{\bsell(v_i)=\ell\}\omega_\ell h_\ell(\bsu_i).
\end{align*}
When every $n_\ell\ge1$ we can write
\begin{align}\label{eq:defhatmu}
\hat\mu = \sum_{\ell=1}^L\omega_\ell\frac{n_\ell}n\hat\mu_\ell
\quad\text{where}\quad
\hat \mu_\ell = \frac1{n_\ell}\sum_{i\in S_\ell}h_\ell(\bsu_i)
\end{align}
is a within-stratum sample mean. 
For the RQMC method we consider,
we show below that $n_\ell\ge1$ whenever $\beta_\ell>2/n$.
We do not need $\Pr(n_\ell>0)=1$ for unbiasedness of $\hat\mu$, but
without it the stratum sample means $\hat\mu_\ell$ are not always well defined.

At this point we can introduce a small generalization.  Since we already 
have a dependence on $\ell$ in $h_\ell$ we could also 
replace $h_\ell(\bsu)=g(\phi_\ell(\bsu))$ by 
$h_\ell(\bsu) = g_\ell( \phi_\ell(\bsu)) = g_\ell(\bsx)$ to allow the function being
applied to $\bsx$ to vary with $\ell$.
\changed{That is $\mu = \sum_{\ell=1}^L\alpha_\ell\e(g_\ell(\bsx)\giv \bsx\sim P_\ell)$.}

\subsection{Sampling the strata}\label{sec:sampleastratum}

The way we sample strata does not affect MC computations
but it makes a difference in RQMC.
Having chosen $\beta_\ell>0$ with $\sum_{\ell=1}^L\beta_\ell=1$,
we let $B_0=0$ and $B_\ell = \sum_{1\le k\le \ell}\beta_k$ 
for $1\le \ell\le L$.
Then given $v$, we select stratum $\ell$ when
$$
B_{\ell-1}\le v<B_\ell.
$$
While $v=1$ has probability zero, for completeness we
set $\bsell(v)=L$ when $v=1$.
Because $v\sim\dunif[0,1]$,
the expected number of samples from stratum $\ell$
is $n\beta_\ell$.  
We will always order the strata so that
$\beta_1\ge\beta_2\ge\cdots\ge\beta_L$.
\changed{This ordering brings a sampling}
advantage in Section~\ref{sec:stratumdiscrep}.
\changed{It is ordinarily the same ordering as the $\alpha_\ell$, though 
we consider exceptions at the end of Section~\ref{subsec:optimal}.}

We will assume that the sampling values $v_1,\dots,v_n$
that we use are `stratified'.  By this we mean that
every interval $I_k=[(k-1)/n,k/n)$ contains precisely
one of $v_1,\dots,v_n$ for $k=1,\dots,n$.
The next proposition shows that \changed{$n_\ell>0$ whenever $\beta_\ell>2/n$.}

\begin{proposition}\label{prop:boundnl}
For $n\ge1$, let $v_1,\dots,v_n$ be stratified.
For $[A,B)\subset[0,1)$, let $n_*$ be the number
of $v_i\in[A,B)$.  Then
\begin{align}\label{eq:nlbound}
\lceil n\beta\rceil -2 \le n_* \le \lfloor n\beta\rfloor+2
\end{align}
where $\beta=B-A$.
\end{proposition}
\begin{proof}
    See Appendix \ref{sec:proof:prop:boundnl}.
\end{proof}

\subsection{Local and star discrepancy}

For $\bsa\in[0,1)^d$, the anchored box $[\bszero,\bsa)$
has volume $\vol([\bszero,\bsa))=\prod_{j=1}^da_j$.
The local discrepancy of $n$ points $\bsx_1,\dots,\bsx_n\in[0,1)^d$ 
at $\bsa$ is
$$
\delta(\bsa) = \frac1n\sum_{i=1}^n\ind\{\bsx_i\in[\bszero,\bsa)\}
-\vol([\bszero,\bsa)).
$$
The star discrepancy of $\bsx_1,\dots,\bsx_n$ is
$$
D_n^* = D_n^*(\bsx_1,\dots,\bsx_n)=\sup_{\bsa\in[0,1)^d}|\delta(\bsa)|.
$$
The Koksma-Hlawka inequality \citep{hick:2014} is that
\begin{align}\label{eq:kh}
\Bigl| \frac1n\sum_{i=1}^nf(\bsx_i)-\int_{[0,1]^d}f(\bsx)\rd\bsx\Bigr|
\le D_n^*(\bsx_1,\dots,\bsx_n) V_{\hk}(f)
\end{align}
where $V_{\hk}$ is the total variation of $f$ in the
sense of Hardy and Krause.

\subsection{Digital nets}

Here we present the defining properties of digital nets.
They exist in integer bases $b\ge2$ \citep{dick:pill:2010}
but we only consider base $2$ here.
For a dimension $d\ge1$, 
elementary intervals in base $2$ are dyadic hyperrectangles of
the form
$$
H(\bsc,\bsk)=\prod_{j=1}^d\Bigl[\frac{c_j}{2^{k_j}},\frac{c_j+1}{2^{k_j}}\Bigr)
$$
for integers $k_j\ge0$ and  $0\le c_j<2^{k_j}$.
We use $|\bsk|=\sum_{j=1}^dk_j$, so $H(\bsc,\bsk)$ has
volume $2^{-|\bsk|}$. 

For integers  $m\ge t\ge0$ the $n=b^m$ points $\bsz_1,\dots,\bsz_n\in[0,1)^d$
are a $(t,m,d)$-net in base $b$ if
$H(\bsc,\bsk)$ contains precisely $nb^{-|\bsk|}=b^{m-|\bsk|}$ of them, whenever
$m-|\bsk|\ge t$.
Smaller $t$ are better.  We always
have $t\le m$ because $t=m$ just means that
all the points are in $[0,1)^d$.

A $(t,d)$-sequence in base $2$ is an infinite sequence
of values $\bsx_i\in[0,1]^d$, with the following property:
any subsequence $\bsx_{a2^r+1},\dots,\bsx_{(a+1)2^r}$
for integers $a\ge0$ and $r\ge m$ is a $(t,r,d)$-sequence
in base $2$.
With a $(t,d)$-sequence we get an infinite sequence
of $(t,m,d)$-nets.  Consecutive pairs of those $(t,m,d)$-nets
yield an infinite sequence of $(t,m+1,d)$-nets of which
pairs form $(t,m+2,d)$-nets
and so on without limit.
We will use digital nets in base $2$ taken from Sobol' sequences
using the direction numbers from \cite{joe:kuo:2008}.

We call a $(t,m,d)$-net in base $2$ stratified
if for each component $j=1,\dots,d$, the values
$x_{1j},\dots,x_{nj}$ are stratified.
A net with $t=0$ is automatically
stratified, but some digital nets with $t>0$ are also stratified.
The $(t,m,d)$-nets from Sobol' sequences are stratified.
The stratification property means we could use the $j$'th
variable in our $(t,m,d)$-net to select $\ell(i)$ for any $j$.
Stratification implies the `projection regularity' property
of \cite{sloa:joe:1994} in which we never get $x_{ij}=x_{i'j}$ for any $i\ne i'$
and any $j$.

The Sobol' points we use are ones that have been 
scrambled. The scrambles 
we consider are either nested uniform scrambles \citep{rtms}
or random linear scrambles with digital shifts \citep{mato:1998:2}.
Scrambling makes each $\bsz_i\sim\dunif[0,1]^d$ individually
while $\bsz_1,\dots,\bsz_n$ remain a $(t,m,d)$-net collectively.

We will estimate $\mu$ by $\hat\mu=(1/n)\sum_{i=1}^nf(\bsz_i)$
where $\bsz_1,\dots,\bsz_n$ are a scrambled $(t,m,s+1)$-net in
base $2$.
This scrambled net integration has several desirable properties
listed in \cite{owen:rudo:2020}:
\begin{compactenum}[\quad\bf1)]
\item If $f\in L_1[0,1]^{s+1}$, then $\e(\hat\mu)=\mu$.
\item If $f\in L_{1+\epsilon}[0,1]^{s+1}$ for some $\epsilon>0$ ,
then $\Pr( \lim_{n\to\infty} \hat\mu=\mu)=1$.
\item If $f\in L_2[0,1]^{s+1}$, then $\var(\hat\mu) =o(1/n)$
as $n\to\infty$.  
\item When every partial derivative of $f$ taken up to 
one time with respect to each component of $\bsz$ is in $L_2$
then $\var(\hat\mu)=O(n^{-3+\epsilon})$.
\item 
There exists a finite $\Gamma$
with $\var(\hat\mu)\le\Gamma\sigma^2/n$ where
$\sigma^2/n$ is the variance under ordinary Monte Carlo (MC)
sampling. 
\end{compactenum}
The scrambled Sobol' nets are stratified. 
\changed{Digital nets in base $2$ have sample sizes $n=2^m$ by definition.  
Arbitrary sample sizes $n$ are available as the first $n$ points of scrambled Sobol' sequences.  Property $1$ holds for any finite $n$ there
and Property 2 holds for any sequence $n\to\infty$ \citep{owen:rudo:2020}. The other properties are specific to $n=2^m$.}

From point 3 above we know that $\var(\hat\mu)=o(n^{-1})$
for $f\in L_2$. 
The Sobol' nets have
$D_n^*=O(n^{-1}\log(n)^{d-1})$ \cite[Theorem 4.10]{nied:1992}.
Scrambled nets are still nets, so this holds for them too.
Then by the Koksma-Hlawka inequality,
$\var(\hat\mu) = O(n^{-2+\epsilon})$ holds for any
$\epsilon>0$ when $V_{\hk}(f)<\infty$. \changed{Then $f$ is of bounded variation in the sense of Hardy and Krause}, which we
describe as $f$ being in BVHK.
For regular $f$ described in point 4 above
we have $\var(\hat\mu)=O(n^{-3+\epsilon})$.
Higher order digital nets can attain even better
convergence rates under correspondingly stronger
regularity \citep{dick:2011}, though they are not
widely used for higher dimensions.
It is common to write these powers of $\log(n)$
as $O(n^\epsilon)$. For a discussion of that
point, see \cite{wherearethelogs}.

While \changed{variance rates} $O(n^{-\rho})$ for $\rho\in\{1,2,3\}$ are widely
studied, these come from error bounds. The actual convergence rate need 
not be near an integer.
\cite{liu:2025} shows examples where certain singular behavior
of the integrand can give any rate between $1$ and $2$.  In
applications, the precise rate of convergence might not be known
before computation begins.

\subsection{Discrepancy within each stratum}\label{sec:stratumdiscrep}

Here we consider how uniformly distributed
the points in each stratum are. 
The first case we consider has $\beta_\ell$
decreasing with $\ell$ and equal to negative powers of $2$.
We show how to find such $\beta_\ell$ in Section~\ref{sec:powtwo}.
Then the sample points within any given stratum
are scrambled nets, and the error can be
as good as $O(n^{-3+\epsilon})$ for regular
integrands.  After that, we consider 
arbitrary $\beta_\ell>2/n$.  Then the points in a given
stratum still have small discrepancy but they are not
generally scrambled digital nets.  In that case the
variance rate in stratum $\ell$ is $O(n_\ell^{-2+\epsilon})$ for 
within-stratum integrands that are in BVHK.

\begin{proposition}\label{prop:stratanets}
  For $n=2^m$ and $s\ge1$, let $\bsz_1,\dots,\bsz_{n}\in[0,1]^{s+1}$ be a
  stratified $(t,m,s+1)$-net in base $2$.
  For $L\ge2$,  let $\beta_\ell = 2^{-\kappa_\ell}$ for $1\le\kappa_1\le\kappa_2\le\cdots\le \kappa_L< m$ satisfy $\sum_{\ell=1}^L\beta_\ell=1$.
  Write $\bsz_i=(v_i,u_{i1},\dots,u_{is})$ for $\bsu_i\in[0,1]^s$. Then for each $\ell=1,\dots,L$, the points $\bsu_i$ for
  which $\bsell(v_i)=\ell$ form a $(t_\ell,m_\ell,s)$-net in base $2$
  with $m_\ell=m-\kappa_\ell$ and $t_\ell=\min(t,m_\ell)$.
\end{proposition}
\begin{proof}
See Appendix~\ref{sec:proof:prop:stratanets}.
\end{proof}

To motivate the constraint that $\beta_1\ge\beta_2\ge\cdots\ge\beta_L$, consider
what would happen if we took $\beta_1=1/8$,
$\beta_2=1/2$, $\beta_3=1/4$ and $\beta_4=1/8$.
Then the points for the second stratum would arise
from $v\in [1/8,5/8)$. This is not an elementary interval
in base $2$, but is instead the union of four such
sets each of length $1/8$. Using the proper ordering
$\beta_1=1/2$, $\beta_2=1/4$, $\beta_3=1/8$ and $\beta_4=1/8$,
we get intervals $[0,1/2)$, $[1/2,3/4)$, $[3/4,7/8)$ and $[7/8,1)$
for the four strata and these are all elementary intervals in base $2$.
The digital net property assures a better equidistribution property
for an elementary interval than it does for an equal sized union
of two or more elementary intervals.

Now we consider what happens when $\beta_\ell$ are
not negative powers of two. We get low discrepancy
within strata but we do not get digital nets there.

\begin{proposition}\label{prop:stilllowdiscrep}
Let $\beta_1,\dots,\beta_L\in(0,1)$ sum to one. 
Let $\bsz_1,\dots,\bsz_n$ be a stratified
$(t,m,s+1)$-net in base $2$ with $n>2/\min_\ell\beta_\ell$.   
Then $n_\ell>0$ and the star discrepancy
of $\{\bsu_i\mid i\in S_\ell\}$ is $O(n_\ell^{-1}\log(n_\ell)^{s})$.
\end{proposition}
\begin{proof}
See Appendix~\ref{sec:proof:prop:stilllowdiscrep}.
\end{proof}

\section{Choosing $\beta_\ell$ for given $\alpha_\ell$}\label{sec:inefficiencies}

Here we develop some principled criteria for choosing $\beta_\ell$.
Our working model is that
\begin{align}\label{eq:varmodel}
\var(\hat\mu_\ell) = \tau_\ell n_\ell^{-\rho}
\end{align}
for some $\rho>0$ and $\tau_\ell>0$.  For simplicity, we assume that
every stratum has the same rate~$\rho$, apart from
a remark at the end of Section~\ref{subsec:optimal}.  The appropriate value of $\rho$
is not necessarily known at the time we choose $\beta_\ell$.
We will show that choosing $\beta_\ell$ based on the wrong
value of $\rho$ can adversely affect the implied constant
in the asymptotic MSE while leaving the convergence rate
unchanged.

Values of $\rho$ equal to $1$, $2$ and $3$
correspond to upper bounds on the convergence
rates for MC, QMC and RQMC variances,
assuming $f\in L^2$, $V_{\hk}(f)<\infty$ and the aforementioned
regularity condition on $f$, respectively, ignoring logarithmic
factors in the latter two cases.  It often happens that the empirical
MSE of an RQMC rule follows an apparent rate of $n^{-\rho}$
for a value of $\rho$ that is not close to an integer.   Unbounded
integrands are not in BVHK, and yet they will have
an MSE that is $o(n^{-1})$. Those often appear to have
$1<\rho<2$.

Our sampling algorithm will deliver $n_\ell$ 
with $|n_\ell -n\beta_\ell|\le2$ by Proposition~\ref{prop:boundnl}.
For $\rho<2$,
we will work as if $n_\ell = n\beta_\ell$, exactly.
Changing one or two points per stratum will make
a difference of $O(1/n)$ to the integral estimate
for a bounded integrand and will have a root mean
squared difference of this magnitude for finite
variances.  The overall root mean squared error (RMSE) is $O(n^{-\rho/2})$
and so for $\rho<2$ these small differences in $n_\ell$
are negligible.

For $\rho>2$, we have the reverse situation. Changing $n_\ell$
by even $\pm1$ can worsen the convergence rate because
changing $\hat\mu$ by $O(1/n)$ is a bigger change
than the RMSE of $O(n^{-\rho/2})$. A convergence
rate with $\rho>2$ for scrambled Sobol' points is only
available for sample sizes that are powers of two.
Section~\ref{sec:powtwo} shows how to \changed{enforce that all} $n_\ell$
\changed{are} equal to powers of two.

Either way, we can assume that $n_\ell=n\beta_\ell$. Then with $\omega_\ell=\alpha_\ell/\beta_\ell$, we can rewrite~\eqref{eq:defhatmu} as
$$
\hat\mu = \sum_{\ell=1}^L\omega_\ell\frac{n_\ell}n\hat\mu_\ell
= \sum_{\ell=1}^L\alpha_\ell\hat\mu_\ell.
$$
The importance sampling ratio $\omega_\ell$ has
canceled from the expression for $\hat\mu$, leaving
a simply weighted sum of within-stratum means.

\subsection{Optimal sample sizes}\label{subsec:optimal}
Under the model~\eqref{eq:varmodel},
$\var(\hat\mu) = \var\bigl(\sum_{\ell=1}^L\alpha_\ell\hat\mu_\ell\bigr)
  \le \bigl(
\sum_{\ell=1}^L\alpha_\ell \tau_\ell^{1/2}n_\ell^{-\rho/2}
  \bigr)^2
$
where the upper bound is from the very conservative assumption that
the estimates within all $L$ strata have unit correlations.
\changed{RQMC points have a space filling property: 
the values of $\bsu_i\in[0,1]^s$ sampled for $\bsell(v_i)=\ell$
tend to fill in the gaps left behind by $\bsu_i$ for $\bsell(v_i)\ne \ell$. Then a zero or even negative correlation among stratum means is more plausible. Under a working model of $0$ inter-stratum correlation, $\var(\hat\mu)\approx \sum_{\ell=1}^L\alpha_\ell^2\tau_\ell/n_\ell^\rho$.}

\changed{Then with no prior information on the relative sizes of $\tau_\ell$ we adopt the criterion
\begin{align}\label{eq:crit0}
C_0(\rho)=C_0(\rho;n_1,\dots,n_L)=\sum_{\ell=1}^L \alpha_\ell^2 n_\ell^{-\rho}.
\end{align}
The conservative criterion assuming unit correlations is $C_1(\rho)=\sum_{\ell=1}^L\alpha_\ell n_\ell^{-\rho/2}$.}

Forming the Lagrangian for~\eqref{eq:crit0} with the constraint $\sum_\ell n_\ell=n$ and setting to zero the partial derivatives
with respect to $n_\ell$ leads to
\begin{align}\label{eq:criterion0}
n_\ell \propto \alpha_\ell^{2/(\rho+1)}.
\end{align} 
We then suppose that 
the user chooses integers $n_\ell$ summing
to $n=2^m$ that are nearly
proportional to $\alpha_\ell^{2/(\rho+1)}$.
The conservative criterion 1 leads to $n_\ell\propto\alpha_\ell^{2/(\rho+2)}$.


\changed{For both criteria}
higher convergence rates $\rho$ tend to equalize the $n_\ell$.
We can think of $n_\ell=n/L$ as arising in
the limit as $\rho\to\infty$. We will also see
equal $n_\ell$ in a minimax argument in Section~\ref{sec:minimax}.
\changed{The conservative criterion 1 favors stratum sizes more nearly equal than under criterion 0 while not exactly equal as the even more conservative minimax argument gives.}

We now have six principled ways to choose how $n_\ell$ should scale
with $\alpha_\ell$, for criteria $0$ and $1$, and three different
hypotheses on $\rho$. They lead to the four different 
rates given in Table~\ref{tab:criterions}.
The case for criterion $0$ and $\rho=1$ (uncorrelated estimates at
the MC rate) is the only one with a proportional
allocation.  The other choices all allocate more than
a proportional number of samples to the
strata with smaller $\alpha_\ell$. This is also true 
for the plain MC rate ($\rho=1$) under the conservative criterion 1.

When $\alpha_1=\cdots=\alpha_L=1/L$, then 
\changed{our criteria} do not distinguish between different values of $\rho$. After raising $\alpha_\ell$ to any power they are still equal and then
a near equal integer allocation is appropriate.

When we do have information about the relative sizes of $\tau_\ell$,
then we can replace $\alpha_\ell$ by $\alpha_\ell \tau_\ell^{1/2}$ in the formulas
in Table~\ref{tab:criterions}.  In that case, the proportionality for $C_0$
and $\rho=1$ matches the Neyman allocation from survey sampling,
in \cite{neym:1934}.  When the cost to sample in stratum $\ell$ is
proportional to $c_\ell>0$ and we constrain the total cost $\sum_{\ell=1}^Lc_\ell n_\ell$, then the Lagrange multiplier argument leads us to replace
$\alpha_\ell$ by $\alpha_\ell \tau_\ell^{1/2} c_\ell^{-1/2}$ in Table~\ref{tab:criterions}.
\changed{When using this extra information, we should order the strata to be in non-increasing order of $\alpha_\ell \tau_\ell^{1/2} c_\ell^{-1/2}$ for either criterion.} 

\begin{table}
\centering
\begin{tabular}{lccc}
\toprule
$\rho:$ & 1 & 2 & 3 \\
\midrule
$C_0$ & $\alpha_\ell$ & $\alpha_\ell^{2/3}$ & $\alpha_\ell^{1/2}$ \\[1ex]
$C_1$ & $\alpha_\ell^{2/3}$ & $\alpha_\ell^{1/2}$ & $\alpha_\ell^{2/5}$\\
\bottomrule
\end{tabular}
\caption{\label{tab:criterions}
This table shows ideal rates for $n_\ell$ as functions of $\alpha_\ell$ \changed{with $C_0$ and $C_1$ for criteria $0$ and $1$ respectively.}
These rates assume the variance within stratum $\ell$
is proportional to $n_\ell^{-\rho}$.
} 
\end{table}

\changed{If $\rho$ varies from stratum to stratum, taking the value $\rho_\ell$ in stratum $\ell$, then we anticipate that larger $\rho_\ell$ will generally lead to smaller $n_\ell$. With better convergence, fewer points should be needed, just as in the common $\rho$ case where smaller $\tau_\ell$ leads to smaller $n_\ell$. We do not have a rigorous proof for the varying $\rho_\ell$ case, but Appendix~\ref{sec:varyrho} gives a heuristic argument.}

\subsection{Suboptimality}
Now suppose that $\var(\hat\mu_\ell)\propto n_\ell^{-\rho}$ but due to imperfect
knowledge, we work as if $\var(\hat\mu_\ell)\propto n_\ell^{-\gamma}$.
We measure inefficiency by the ratio of the MSE we get if we design with a value $\gamma$ when we should
have used $\rho$, so smaller values are better, and $1$ is best possible.

Under criterion 0, we would then 
take $n_\ell \propto \alpha_\ell^{2/(\gamma+1)}$.
The inefficiency of using $\gamma$ given that we should have
used $\rho$ is
\begin{align}\label{eq:i0ineff}
I_0(\gamma\giv \rho)
&=
\frac{\sum_{\ell=1}^L\alpha^2_\ell \Bigl(
n \alpha_\ell^{2/(\gamma+1)}\bigm/\sum_{k=1}^L\alpha_k^{2/(\gamma+1)}\Bigr)^{-\rho}
}{\sum_{\ell=1}^L\alpha^2_\ell \Bigl(
n \alpha_\ell^{2/(\rho+1)}\bigm/\sum_{k=1}^L\alpha_k^{2/(\rho+1)}\Bigr)^{-\rho}}
\notag\\
&=
\sum_{\ell=1}^L\alpha_\ell^{2-2\rho/(\gamma+1)} 
\frac{\bigl(\sum_{k=1}^L\alpha_k^{2/(\gamma+1)}\bigr)^\rho}
{\bigl(\sum_{k=1}^L\alpha_k^{2/(\rho+1)}\bigr)^{\rho+1}}.
\end{align}
The factor $n^{-\rho}$ has canceled from the numerator and denominator
and so we get the same rate of convergence with an incorrect $\gamma$
as we would with the optimal $\gamma=\rho$. Only the lead constant
is affected.
\changed{The analogous formula for criterion 1 is
\begin{align*}
I_1(\gamma\giv \rho)
&=\frac{
\bigl(\sum_{\ell=1}^L\alpha_\ell^{1-\rho/(\gamma+2)} \bigr)^2
\bigl(\sum_{k=1}^L\alpha_k^{2/(\gamma+2)}\bigr)^\rho}
{\bigl(\sum_{k=1}^L\alpha_k^{2/(\rho+2)}\bigr)^{\rho+2}}.
\end{align*}
}

\changed{Under criterion 0}
a cautious choice is to take $n_\ell\propto \alpha_\ell^{2/(\gamma_0+1)}$ where
\begin{align}\label{eq:minmax0}
\gamma_0 = \argmin_{1\le\gamma\le3} \max_{1\le\rho\le 3}I_0(\gamma\giv\rho).
\end{align}
When we know that $\rho\le2$ we might only let $\rho$ and $\gamma$ belong to $[1,2]$. 


The minimax approach of equation~\eqref{eq:minmax0} 
gives us ways to choose $\gamma$ for a given list
of values $\alpha_1,\dots,\alpha_L$.
It often happens that the worst value of $\rho$
for given $\gamma$ is at one of the allowed
extremes, but this does not always happen.
We have seen such exceptions when one of the $\alpha_\ell$
is much larger than the others.  In our numerical examples the maximum over $\rho$ of $I_0(2\giv\rho)$ or $I_0(3\giv\rho)$ was not much larger
than $I_0(\gamma_0\giv\rho)$.

Next we consider recommendations that can be made
generally without specific values of $\alpha_\ell$ in mind.
The first thing to notice is that both $I_0$ and $I_1$ contain
a numerator sum that could potentially have a negative
power of $\alpha_\ell$.  For general use where some $\alpha_\ell$
might be small, we want to avoid that possibility.  There 
is no possibility of an even more negative power in the denominator to compensate.
To keep a positive exponent in the numerator of $I_0$ we must have
$2-2\rho/(\gamma+1)\ge0$.
After rearrangement this becomes a constraint
that $\gamma\ge \rho-1$, whatever the unknown $\rho$
might be. When $\rho\doteq 3$ is a possibility,
then we should take $\gamma\ge2$.  
If instead we use $\gamma=1$ in a $\rho=3$ setting
then the first factor in $I_0$ is greater than $\alpha_L^{-1}$.
For $I_1$ we want $1-\rho/(\gamma+2)\ge0$
or $\gamma\ge\rho-2$.  Since we already had $\gamma\ge\rho-1$
using criterion 0, this is not a further constraint.
When we know that $\rho\le2$, then any $\gamma\ge1$
avoids a negative exponent for $\alpha_\ell$.

When $\rho=1$ (MC rate), then
\begin{align*}
I_0(2\giv 1) &= \sum_{\ell=1}^L\alpha_\ell^{4/3}\sum_{\ell=1}^L\alpha_\ell^{2/3},\quad\text{and}\\
I_0(3\giv 1) &= \sum_{\ell=1}^L\alpha_\ell^{3/2}\sum_{\ell=1}^L\alpha_\ell^{1/2}.
\end{align*}
It is not difficult to show that $I_0(3\giv1)\ge I_0(2\giv1)\ge I_0(1\giv1)=1$.
Proposition~\ref{prop:monotone} shows that more general monotonicities hold.

\begin{proposition}\label{prop:monotone}
For $L>1$, let $\alpha_\ell>0$ for $\ell=1,\dots,L$ sum to one.
If $\rho\ge1$ then
\begin{align}\label{eq:monosubopt}
I_0(\gamma_1\giv\rho)\ge I_0(\gamma_2\giv\rho)\ge I_0(\rho\giv\rho)
\end{align}
whenever either $\gamma_1\ge\gamma_2\ge\rho$
or $1\le\gamma_1\le\gamma_2\le\rho$.
\end{proposition}
\begin{proof}
See Appendix~\ref{sec:proof:prop:monotone}.
\end{proof}

If the unknown true rate is $\rho=2$, we are better off
choosing $\gamma=3$ than $\gamma=1$. Proposition~\ref{prop:beton3vs1whenits2}
has this `bet on optimism' result.

\begin{proposition}\label{prop:beton3vs1whenits2}
For any $\alpha_\ell>0$ with $\sum_{\ell=1}^L\alpha_\ell=1$,
$$
I_0(3\giv 2)\le I_0(1\giv 2).
$$
\end{proposition}
\begin{proof}
See Appendix~\ref{sec:proof:prop:beton3vs1whenits2}.
\end{proof}

\section{Making each $n_\ell$ a power of two}\label{sec:powtwo}

Here we consider how to allocate sample sizes $n_\ell=2^{m_\ell}$ that 
sum to a given total sample size $n=2^m$.
We assume throughout that $n\ge L$ and that we
want every $n_\ell\ge1$. 
The advantage of having every $n_\ell$ be a power of two is that
we then get scrambled net sampling within each of the $L$
strata as shown by Proposition~\ref{prop:stratanets}.  Sample sizes that are a power of two are a best practice
in scrambled Sobol' sampling \citep{firstsobol}.  The
results with MSE $O(n^{-3+\epsilon})$ for scrambled Sobol'
points require such sample sizes. 

A potential non-asymptotic disadvantage of having $n_\ell$ be a power of 
$2$ is that we have many fewer sampling ratios 
$\alpha_\ell/\beta_\ell$ at our disposal.
For example, with $L=2$, we would have to use $n_1=n_2=n/2$
and this could be a bad fit for $\alpha_1=0.9999$ and $\alpha_2=0.0001$.

We take $\beta_\ell=2^{-\kappa_\ell}$ such
that $\sum_{\ell=1}^L\beta_\ell=1$.
That is, $\beta_\ell$ are a partition of $1$.
Partitions like $(1/2,1/4,1/4)$ and $(1/4,1/2,1/4)$
are equivalent because any assignment of one of those to $\beta_\ell$
can be achieved by the other after a permutation of indices.

We suppose that $\beta_1\ge\beta_2\ge\cdots\ge\beta_L$ for $L\ge2$
and look at sample fractions $(n_1/n,n_2/n,\dots,n_L/n)$.
For $L=2$ the only solution is $(1/2,1/2)$ and
for $L=3$ the only solution is $(1/2,1/4,1/4)$.
For $L=4$, there are two possibilities: $(1/2,1/4,1/8,1/8)$ and $(1/4,1/4,1/4,1/4)$
and for $5\le L\le 10$ the numbers of
possibilities are: 3, 5, 9, 16, 28 and 50.
The tables in Figure~\ref{tab:partitions} show the partitions for $4\le L\le 8$.

\begin{figure}[t!]
  \centering
  \makebox[\textwidth][c]{%
    \begin{minipage}[t]{0.48\textwidth}
      \vspace{0pt}
      \centering
      \begin{tabular}{lccccccccc}
        \toprule
        L & \multicolumn{4}{l}{Partitions:} & & & & & \\
        \midrule
        4 & $\tfrac12$ & $\tfrac14$ & $\tfrac18$  & $\tfrac18$ & & & & & \\[0.8ex]
          & $\tfrac14$ & $\tfrac14$ & $\tfrac14$  & $\tfrac14$ & & & & & \\[0.8ex]
        \midrule
        5 & $\tfrac12$ & $\tfrac14$ & $\tfrac18$  & $\tfrac1{16}$ & $\tfrac1{16}$ & & & & \\[0.8ex]
          & $\tfrac12$ & $\tfrac18$ & $\tfrac18$  & $\tfrac18$     & $\tfrac18$     & & & & \\[0.8ex]
          & $\tfrac14$ & $\tfrac14$ & $\tfrac14$  & $\tfrac18$     & $\tfrac18$     & & & & \\[0.8ex]
        \midrule
        6 & $\tfrac12$ & $\tfrac14$ & $\tfrac18$  & $\tfrac1{16}$ & $\tfrac1{32}$ & $\tfrac1{32}$ & & & \\[0.8ex]
          & $\tfrac12$ & $\tfrac14$ & $\tfrac1{16}$ & $\tfrac1{16}$ & $\tfrac1{16}$ & $\tfrac1{16}$ & & & \\[0.8ex]
          & $\tfrac12$ & $\tfrac18$ & $\tfrac18$  & $\tfrac18$     & $\tfrac1{16}$ & $\tfrac1{16}$ & & & \\[0.8ex]
          & $\tfrac14$ & $\tfrac14$ & $\tfrac14$  & $\tfrac18$     & $\tfrac1{16}$ & $\tfrac1{16}$ & & & \\[0.8ex]
          & $\tfrac14$ & $\tfrac14$ & $\tfrac18$  & $\tfrac18$     & $\tfrac18$     & $\tfrac18$     & & & \\[0.8ex]
        \midrule
        7 & $\tfrac12$  & $\tfrac14$  & $\tfrac18$  & $\tfrac1{16}$ & $\tfrac1{32}$ & $\tfrac1{64}$ & $\tfrac1{64}$ & & \\[0.8ex]
          & $\tfrac12$  & $\tfrac14$  & $\tfrac18$  & $\tfrac1{32}$ & $\tfrac1{32}$ & $\tfrac1{32}$ & $\tfrac1{32}$ & & \\[0.8ex]
          & $\tfrac12$  & $\tfrac14$  & $\tfrac1{16}$ & $\tfrac1{16}$ & $\tfrac1{16}$ & $\tfrac1{32}$ & $\tfrac1{32}$ & & \\[0.8ex]
          & $\tfrac12$  & $\tfrac18$  & $\tfrac18$  & $\tfrac18$     & $\tfrac1{16}$ & $\tfrac1{32}$ & $\tfrac1{32}$ & & \\[0.8ex]
          & $\tfrac12$  & $\tfrac18$  & $\tfrac18$  & $\tfrac1{16}$ & $\tfrac1{16}$ & $\tfrac1{16}$ & $\tfrac1{16}$ & & \\[0.8ex]
          & $\tfrac14$  & $\tfrac14$  & $\tfrac14$  & $\tfrac18$     & $\tfrac1{16}$ & $\tfrac1{32}$ & $\tfrac1{32}$ & & \\[0.8ex]
          & $\tfrac14$  & $\tfrac14$  & $\tfrac14$  & $\tfrac1{16}$ & $\tfrac1{16}$ & $\tfrac1{16}$ & $\tfrac1{16}$ & & \\[0.8ex]
          & $\tfrac14$  & $\tfrac14$  & $\tfrac18$  & $\tfrac18$     & $\tfrac18$     & $\tfrac1{16}$ & $\tfrac1{16}$ & & \\[0.8ex]
          & $\tfrac14$  & $\tfrac18$  & $\tfrac18$  & $\tfrac18$     & $\tfrac18$     & $\tfrac18$     & $\tfrac18$     & & \\[0.8ex]
        \bottomrule
      \end{tabular}
    \end{minipage}%
    \hspace{0.08\textwidth}%
    \begin{minipage}[t]{0.52\textwidth}
      \vspace{0pt}
      \centering
      \begin{tabular}{l*{8}{c}}
        \toprule
        L & \multicolumn{8}{c}{Partitions:} \\
        \midrule
        8  & $\tfrac12$ & $\tfrac14$ & $\tfrac18$   & $\tfrac1{16}$ & $\tfrac1{32}$ & $\tfrac1{64}$  & $\tfrac1{128}$ & $\tfrac1{128}$ \\[0.8ex]
           & $\tfrac12$ & $\tfrac14$ & $\tfrac18$   & $\tfrac1{16}$ & $\tfrac1{64}$ & $\tfrac1{64}$  & $\tfrac1{64}$   & $\tfrac1{64}$   \\[0.8ex]
           & $\tfrac12$ & $\tfrac14$ & $\tfrac18$   & $\tfrac1{32}$ & $\tfrac1{32}$ & $\tfrac1{32}$  & $\tfrac1{64}$   & $\tfrac1{64}$   \\[0.8ex]
           & $\tfrac12$ & $\tfrac14$ & $\tfrac1{16}$ & $\tfrac1{16}$ & $\tfrac1{16}$ & $\tfrac1{32}$  & $\tfrac1{64}$   & $\tfrac1{64}$   \\[0.8ex]
           & $\tfrac12$ & $\tfrac14$ & $\tfrac1{16}$ & $\tfrac1{16}$ & $\tfrac1{32}$ & $\tfrac1{32}$  & $\tfrac1{32}$   & $\tfrac1{32}$   \\[0.8ex]
           & $\tfrac12$ & $\tfrac18$ & $\tfrac18$   & $\tfrac18$    & $\tfrac1{16}$ & $\tfrac1{32}$  & $\tfrac1{64}$   & $\tfrac1{64}$   \\[0.8ex]
           & $\tfrac12$ & $\tfrac18$ & $\tfrac18$   & $\tfrac18$    & $\tfrac1{32}$ & $\tfrac1{32}$  & $\tfrac1{32}$   & $\tfrac1{32}$   \\[0.8ex]
           & $\tfrac12$ & $\tfrac18$ & $\tfrac18$   & $\tfrac1{16}$ & $\tfrac1{16}$ & $\tfrac1{16}$  & $\tfrac1{32}$   & $\tfrac1{32}$   \\[0.8ex]
           & $\tfrac12$ & $\tfrac18$ & $\tfrac1{16}$ & $\tfrac1{16}$ & $\tfrac1{16}$ & $\tfrac1{16}$  & $\tfrac1{16}$   & $\tfrac1{16}$   \\[0.8ex]
           & $\tfrac14$ & $\tfrac14$ & $\tfrac14$   & $\tfrac18$    & $\tfrac1{16}$ & $\tfrac1{32}$  & $\tfrac1{64}$   & $\tfrac1{64}$   \\[0.8ex]
           & $\tfrac14$ & $\tfrac14$ & $\tfrac14$   & $\tfrac18$    & $\tfrac1{32}$ & $\tfrac1{32}$  & $\tfrac1{32}$   & $\tfrac1{32}$   \\[0.8ex]
           & $\tfrac14$ & $\tfrac14$ & $\tfrac14$   & $\tfrac1{16}$ & $\tfrac1{16}$ & $\tfrac1{16}$  & $\tfrac1{32}$   & $\tfrac1{32}$   \\[0.8ex]
           & $\tfrac14$ & $\tfrac14$ & $\tfrac18$   & $\tfrac18$    & $\tfrac18$    & $\tfrac1{16}$  & $\tfrac1{32}$   & $\tfrac1{32}$   \\[0.8ex]
           & $\tfrac14$ & $\tfrac14$ & $\tfrac18$   & $\tfrac18$    & $\tfrac1{16}$ & $\tfrac1{16}$  & $\tfrac1{16}$   & $\tfrac1{16}$   \\[0.8ex]
           & $\tfrac14$ & $\tfrac18$ & $\tfrac18$   & $\tfrac18$    & $\tfrac18$ & $\tfrac18$  & $\tfrac1{16}$   & $\tfrac1{16}$   \\[0.8ex]
           & $\tfrac18$ & $\tfrac18$ & $\tfrac18$   & $\tfrac18$    & $\tfrac18$    & $\tfrac18$     & $\tfrac18$      & $\tfrac18$      \\[0.8ex]
        \bottomrule
      \end{tabular}
    \end{minipage}%
  }
  \caption{\label{tab:partitions}
    These are the partitions of unity into $L$ negative powers of $2$
    for $L\in\{4,5,6,7,8\}$.
  }
\end{figure}

The number of such partitions as a function of $L$ is given by sequence A002572 at the online encyclopedia of integer sequences, \citep{oeis}
along with motivating
applications such as counting the number of rooted
binary trees or, more recently, counting nonequivalent compact Huffman codes 
(\cite{elsh:heub:kren:2024}).
The latter reference mentions that the number of solutions
is asymptotic to $R\theta^L$ for $R\doteq 0.14$ and $\theta\doteq1.794$.
This is a somewhat faster growth rate than the Fibonacci numbers have:
their counterpart to $\theta$ is $(1+\sqrt{5})/2\doteq 1.618$.
For $L=10$, the approximation is close to 48.3. 

To choose $\beta_\ell$, or equivalently $n_\ell$,
we use a forward stepwise algorithm.
We initialize every $m_\ell=0$ so every
$n_\ell = 2^{m_\ell}$ is initially $1$.
At every step of the algorithm there are
$n' = n-\sum_{\ell=1}^Ln_\ell$ samples left to allocate.
If $n'=0$, the algorithm terminates.  If $n'>0$ then we 
select one of the indices $\ell$ with $n_\ell \le n'$ and 
double that $n_\ell$, replacing $m_\ell$ by $m_\ell+1$.  Before
describing which $n_\ell$ to double, we show that
any such algorithm always terminates.

\begin{proposition}\label{prop:alwaysterminates}
The above algorithm always terminates in a finite
number of steps.
\end{proposition}
\begin{proof}
    See Appendix \ref{sec:proof:prop:alwaysterminates}.
\end{proof}

Under criterion 0, the ideal proportion of samples in
stratum $\ell$ is
$$\xi_{\ell,\rho}=\alpha_\ell^{2/(\rho+1)}\Bigm/\sum_{k=1}^L\alpha_k^{2/(\rho+1)}.$$
We choose to double the sample size for stratum
$$\ell^*_0 = \argmax_{1\le \ell\le L}\frac{\xi_{\ell,\rho}}{2^{m_{\ell}}}\ind\Bigl\{2^{m_{\ell}}\le n-\sum_{j=1}^L2^{m_j}\Bigr\}$$ 
which we consider to be the most `under-represented'
of the eligible strata.
The analogous algorithm under criterion $1$ is to double the sample size of stratum
$$\ell^*_1 = \argmax_{1\le \ell\le L}\frac{\xi_{\ell,\rho}}{2^{m_{\ell}}}\ind\Bigl\{2^{m_{\ell}}\le n-\sum_{j=1}^L2^{m_j}\Bigr\}$$ 
where
$$\xi_{\ell,\rho}=\alpha_\ell^{2/(\rho+2)}\Bigm/\sum_{k=1}^L\alpha_k^{2/(\rho+2)}.$$
 The pseudo-code for forward allocation algorithm for criterion 0 is given in Algorithm~\ref{alg:fsa}.

\begin{algorithm}[t]
\caption{Forward Stratified Allocation\label{alg:fsa}}
\begin{algorithmic}[1]
\INPUT Budget $n = 2^m\ge L$, stratum weights $\alpha_1, \dots, \alpha_L$, parameter $\rho\ge1$
\OUTPUT Powers $\bm{m} = (m_1, \dots, m_L)$ such that $n_\ell = 2^{m_\ell}$ and $\sum_{\ell=1}^L n_\ell = n$

\STATE Initialize $m_\ell \leftarrow 0$ for all $\ell = 1, \dots, L$
\STATE Compute $\xi_{\ell,\rho} \leftarrow \alpha_\ell^{2/(\rho+1)} \big/ \sum_{k=1}^L \alpha_k^{2/(\rho+1)}$ for all $\ell$
\WHILE{$n - \sum_{\ell=1}^L 2^{m_\ell} > 0$}
    \STATE Find $\ell^* \leftarrow \argmax_{1 \le \ell \le L} \left\{ \dfrac{\xi_{\ell,\rho}}{2^{m_\ell}} \cdot \ind\left\{ 2^{m_\ell} \le n - \sum_{j=1}^L 2^{m_j} \right\} \right\}$
    \STATE Update $m_{\ell^*} \leftarrow m_{\ell^*} + 1$
\ENDWHILE
\RETURN $\bm{m} = (m_1, \dots, m_L)$
\end{algorithmic}
\end{algorithm}

\section{Minimax choice}\label{sec:minimax}
Here we show a minimax property of near-equal allocation.
We do not assume equality among the $\tau_\ell$.
For any integer $N\ge L$, let
\begin{align*}
\Delta_N &=\biggl\{\bsn=(n_1,\dots,n_L)\in \natu^{L}\Bigm| \sum_{\ell=1}^L n_{\ell} = N\biggr\},\quad\text{and}\quad\\
\mathcal{T}&=\Bigl\{\bstau=(\tau_1,\dots,\tau_L)\in [0,\infty)^L\bigm| \max_{1\leq \ell \leq L}\tau_{\ell}>0\Bigr\}
\end{align*}
be spaces of feasible sample sizes and possible variance
constants, respectively. We allow some $\tau_\ell=0$ but
there must be at least one $\tau_\ell>0$ to make the problem
interesting.

We consider the increase in variance that we get
using sample sizes $\bsn$ compared to some
better sample sizes $\tilde\bsn$ that we might
otherwise have used.
For criteria $0$ and $1$ we define suboptimality ratios
\[R_0(\bsn\giv\tilde{\bsn};\bstau) := \frac{\sum_{\ell} \alpha_\ell^2 \tau_{\ell} n_{\ell}^{-\rho}}{\sum_{\ell} \alpha_\ell^2 \tau_{\ell} \tilde{n}_{\ell}^{-\rho}}
 \quad \text{and} \quad R_1(\bsn\giv\tilde{\bsn};\bstau) := \frac{\sum_{\ell} \alpha_\ell \tau_{\ell}^{1/2} n_{\ell}^{-\rho/2}}{\sum_{\ell} \alpha_\ell \tau_{\ell}^{1/2} \tilde{n}_{\ell}^{-\rho/2}}
\]
that compare different sample sizes taking account
of the stratum variances $\tau_\ell$.
 These ratios depend on $\rho$ and $\alpha_1,\dots,\alpha_L$, but we leave those out of the notation
because they are held constant in this section.

\begin{proposition}\label{prop:minmax}
 For both criteria $a \in \{0,1\}$,
    \begin{equation}\label{eq:mainpb}
    \min_{\bsn\in \Delta_N}\max_{\tilde\bsn\in \Delta_N}\max_{\bstau\in \mathcal{T}} R_a(\bsn\giv\tilde\bsn;\bstau)=R_a(\bm{n}^*\giv\tilde{\bsn}^*;\bstau^*)
    \end{equation}
for minimax sample sizes
    $$n_{\ell}^*=\begin{cases}
        \left\lfloor{\frac{N}{L}}\right\rfloor +1, & 
        \textrm{for }\ell=1,\dots,r\\
        \left\lfloor{\frac{N}{L}}\right\rfloor, &\textrm{for }\ell=r+1,\dots,L
    \end{cases}$$
    with $r=N-L\lfloor N/L \rfloor$,
worst case within-stratum 
variances $\tau_\ell=\lambda \ind\{\ell = \ell_*$\}
for any $\lambda>0$ and worst case alternative
sample sizes
\begin{equation*}
\begin{array}{ll}
\tilde{n}_{\ell}^* =
\begin{cases}
    N-L+1, & \text{if } \ell = \ell^* \\
    1, & \text{if } \ell \in\{1,2,\dots,L\}\setminus\{\ell^*\}
\end{cases}
\end{array}
\end{equation*}
for any $\ell^*\in\{r+1,\dots,L\}$.
\end{proposition}
\begin{proof}
See Appendix \ref{sec:proof:prop:minmax}
\end{proof}


When we must have $n_\ell=2^{m_\ell}$ for integers
$m_\ell\ge0$, there is a corresponding minimax
result that takes the $n_\ell$ as nearly equal as
possible.  
Let $r = \lceil \log_2(L)\rceil$
and then every $\beta_\ell$ is either $2^{-r}$ or $2^{1-r}$.
Solving
$s2^{1-r}+(L-s)2^{-r}=1$, we see that there should be $s=2^r-L$ strata with the larger
allocation.
Then we take $n_\ell=n\beta_\ell$ for
$$\beta_\ell =
\begin{cases}
 2^{1-r}, &1\le\ell\le 2^r-L\\
2^{-r}, &2^r-L<\ell\le L.
\end{cases}
$$

\section{A toy example}\label{sec:toy}

Here we consider a very small example with
only one continuous variable and $8$ strata
for illustration.
In stratum $\ell$, the parameter $\theta$
takes the value $\theta_\ell$
and then $x_1\sim\dnorm(\theta_\ell,1)$. This example
uses
\begin{equation}\begin{split}
\label{eq:toyparams}
\alpha &=(0.50,0.44,0.01,0.01,0.01,0.01,0.01,0.01),\quad\text{and}\\
\theta&=(0.7,1.0,1.5,1.6,1.7,1.8,1.9,2.0).
\end{split}
\end{equation}
The within-stratum integrand is
$g(x_1)=\exp(-x_1^2)\cos(x_1)$.
The integrand $f(\bsz)$ for $\bsz\in[0,1]^2$ has a first
component $z_1$ that selects a stratum $\ell$ as described in
Section~\ref{sec:sampleastratum}. Then $x_1=\theta_\ell+\Phi^{-1}(z_2)$
where $\Phi(\cdot)$ is the cumulative distribution function
of the $\dnorm(0,1)$ distribution.  We chose this function because
it is regular enough that we expect to see an RQMC variance of
$O(n^{-3})$. We chose these values of $\alpha_\ell$ so that the forward algorithm would give different $n_\ell$ for $\rho=2$ than it does for $\rho=3$.

We compare the variance of these unbiased estimators:
\begin{enumerate}[1)]
    \item A plain Monte-Carlo estimator: 
$\hat{\mu}_{\mc}=(1/n)\sum_{i=1}^n f(\bsz_i)$
where $\bsz_i\simiid \dustd(0,1)^2$.
    \item A plain RQMC estimator:
$ \hat\mu_{\rqmc}=(1/n)\sum_{i=1}^nf(\bsz_i)$
where $\bsz_1,\dots,\bsz_n$ are the first $n$ points of a scrambled Sobol' sequence in $[0,1]^2$.
   \item An importance adjusted estimator:
$\hat\mu_{\rqmc(\adj)} =(1/n)\sum_{i=1}^n f(\bsz_i)\omega(\ell(i))$
with RQMC points where the $n_\ell$ come from rounding equation~\eqref{eq:criterion0}, using $\rho=2$, because that is the rate
from the Koksma-Hlawka bound for $n_\ell$ not equal to powers of two. 
\item A reweighted RQMC estimator:
$\hat\mu_{\rqmc(2)} =(1/n)\sum_{i=1}^n f(\bsz_i)\omega(\ell(i))$,
where the $n_\ell$ are powers of $2$ chosen by Algorithm~\ref{alg:fsa} with $\rho=3$,
the variance rate for such $n_\ell$.
\item A per stratum estimator:
$\hat\mu_{\rqmc(L)}=\sum_{\ell=1}^L\alpha_\ell\hat\mu_\ell$
where $\hat\mu_\ell$ are independently sampled RQMC estimates,
with the same $n_\ell$ as in $\hat\mu_{\rqmc(2)}$.
\end{enumerate}
For each estimator we compute $R=500$ replicates at sample
sizes $n=2^m$ for $3\le m\le 12$ and then graph the
sample variances. We used nested uniform scrambling \citep{rtms} to do the scrambling.

The results are shown in Figure~\ref{fig:toyvariances}.
The smaller sample sizes are comparable to the number of
strata and so when we look for an understanding of convergence
we emphasize the larger values of $n$. All of the methods are
unbiased even though some of them have some empty strata
at small $n$.

The variance for Monte Carlo sampling follows
the expected $O(1/n)$ rate.
Simply replacing the MC points by RQMC points
improves the rate to roughly $O(1/n^2)$.
Adjusting the sampling fractions from $\alpha_\ell$ 
to $\beta_\ell$ using $\rho=2$ 
retains the RQMC rate and provides lower variance
at larger samples sizes. 
Using $\rho=3$ to choose sample sizes $n_\ell$, we compared
two methods.  The first generates the estimate $\hat\mu_{\rqmc(L)}$ using $L$ independent RQMC samples, with $n_\ell$ points in stratum $\ell$.
The other uses the combined RQMC estimate $\hat\mu_{\rqmc(2)}$.
These latter two both give an apparent $O(n^{-3})$
rate for the variance.  The combined estimate has a
better constant in this example. See our discussion
of between strata correlations in Section~\ref{sec:discussion}.

\begin{figure}[t]
    \centering
    \includegraphics[width=0.9\linewidth]{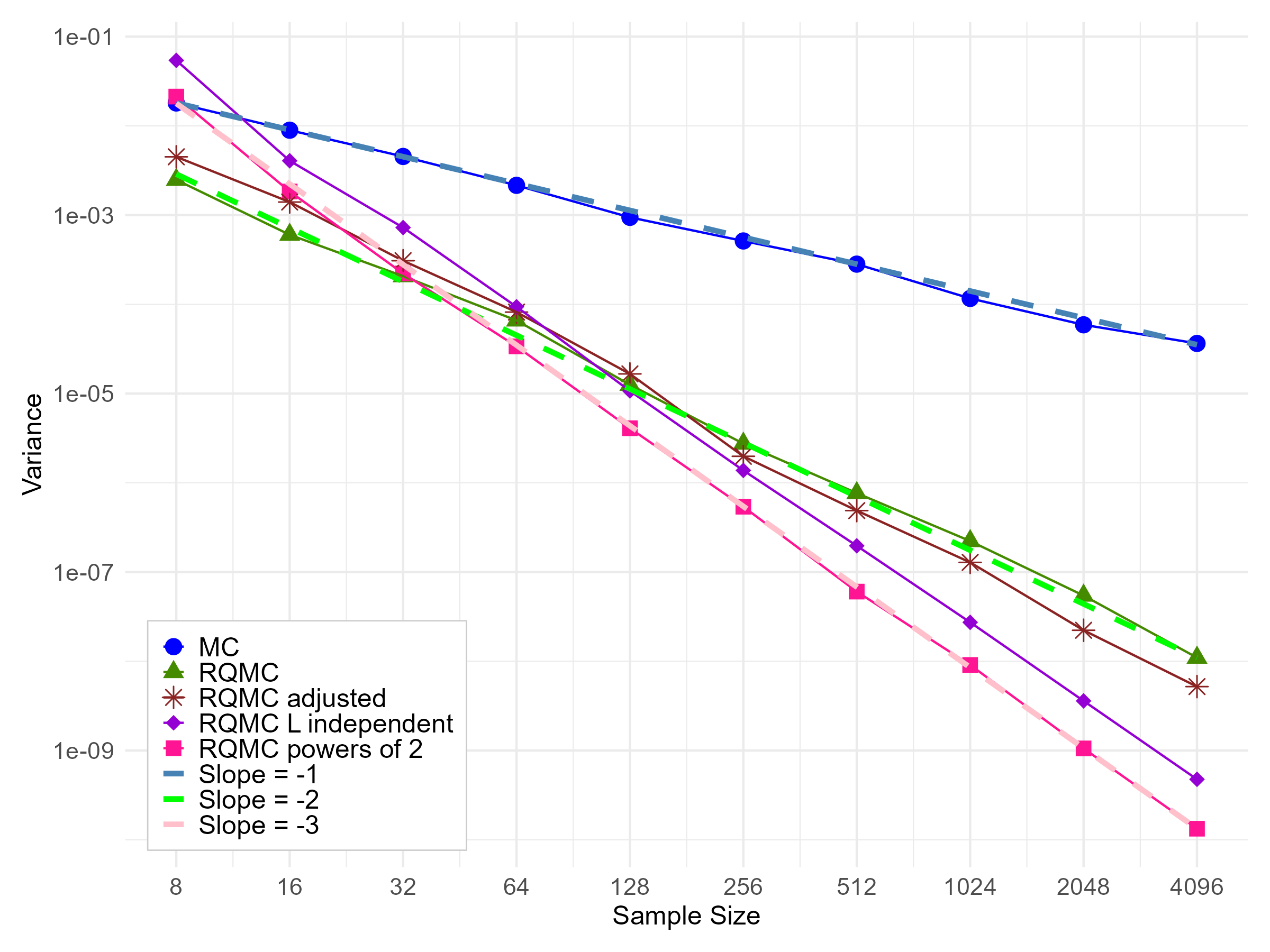}
    \caption{Variance of the 5 estimators for the estimation of $\e(f(\bsz))$ versus the sample size $n$, for the toy integrand. \changed{The top to bottom ordering of the curves in the legend is the same as they have at $n=4096$.} Four of the curves are nearly parallel to reference lines of a given slope.}
    \label{fig:toyvariances}
\end{figure}

We also compared designs for different values of $\rho$
when taking $n_\ell$ to be powers of $2$.
All of the MSEs were at the $n^{-3}$ rate.
The best constant was from using $\rho=3$ and
the worst was from $\rho=\infty$ which makes the
$n_\ell$ equal.  The variances for $\rho=1$ and
$\rho=2$ were the same because those used
the same values of $n_\ell=2^{m_\ell}$.

We noticed an unplanned anomaly in this example because we have $L=8$. 
Then using $\rho=\infty$,
we get $n_\ell =n/8$, which is always a power of $2$ whether we use Algorithm~\ref{alg:fsa}
or just equation~\eqref{eq:criterion0}. As a result, using $\rho=\infty$ always gave $O(n^{-3})$ convergence
where the best we could have had for $L=7$ or $9$ without
forcing $n_\ell$ to be powers of two, would have been $O(n^{-2})$.

\section{Saint-Venant flood model}\label{sec:flood}

We are motivated by engineering problems with substantial
uncertainty.  Great uncertainty often arises from weather
variables or customer behavior.  Our example is the Saint-Venant
flood model from \cite{limb:dero:2010}.  A more precise flood model would be based
on Navier-Stokes equations that have to be solved numerically.  
The Saint-Venant model is a simplification
that we think is well suited for illustrative purposes.

The quantity of interest is the depth $H$ in meters of river water at a
downstream location. The model has
\begin{align}\label{eq:flooddepth}
H = \Bigl( \frac{Q}{K_sB\sqrt{\alpha}}\Bigr)^{3/5},
\end{align}
where $Q$ is the flow rate (meters$^3$/second), $B$ is the width of the river (meters), $K_s$ is the Strickler
coefficient (meters$^{1/3}$/second) quantifying the roughness of
the river bed and $\alpha$ is the (dimensionless) slope of the river bed.
It is given by $\alpha = (Z_m-Z_v)/D$ where $Z_m$ and $Z_v$ are the
upstream and downstream heights of the river in meters, and $D$ denotes the river's length in meters.

\begin{table}[t]
\centering
\begin{tabular}{lll}
\toprule
Variable & Nominal & Adverse\\
\midrule
$Q$ &  Fr\'echet (shape $6$, scale $1300$) &Fr\'echet (shape $6$, scale $3900$)\\
$K_s$ & Gamma (shape $90$, scale $1/3$)& Gamma (shape $15$, scale $1$)\\
$Z_v$ & Uniform ($49$,$51$) & none\\
$Z_m$ & Uniform ($54$,$56$)& none\\
\bottomrule
\end{tabular}
\caption{\label{tab:distns}Here are nominal distributions
for four variables in the Saint-Venant model. Two of  
them also have adverse distributions that bring worse flooding.}
\end{table}

The inputs to $H$ are random and prior
researchers sampled each input independently from suitable
distributions.  Those prior contexts did not have mixtures. 
We have modified those distributions to include mixtures
of nominal and adverse distributions.  Our strata involve 
different subsets of the variables being adversely sampled.

This model is usually given with a symmetric triangular
density for $D$ over the interval $[4990,5010]$.
It is usually seen that uncertainty in $D$ makes
very little difference. See for example \cite{ioos:lema:2015}. Therefore we fix $D=5000$.
Similarly, $B$ is commonly given a symmetric triangular
distribution over $[295,305]$ and we fix it at $B=300$.
We consider four distributions for the inputs $Q$, $K_s$, $Z_v$ and $Z_m$.
With probability $0.95$ all four random variables are independently
distributed with nominal distributions given in Table~\ref{tab:distns}.
With probability $0.02$, the distribution of $Q$ takes the adverse
distribution in Table~\ref{tab:distns} while the other variables
are nominal. With probability $0.02$, the distribution of $K_s$ is
adverse and the others are nominal.  With probability $0.01$ both
$Q$ and $K_s$ are adverse. 

The distribution of $K_s$ is customarily given as
$\dnorm(30,7.5^2)$ and then truncated to a range such
as $(0,\infty)$ in \cite{limb:dero:2010} or $(15,\infty)$ in \cite{ioos:2011}.
We have chosen a nominal Gamma distribution, so we do not need
any truncation.  That Gamma distribution has mean
$30$ but a lower variance than $7.5^2$.
The adverse distribution for the Gamma distribution takes lower
values which bring more flooding as $K_s$ appears in
the denominator of $H$.  The distribution of $Q$ is
customarily given as Gumbel with mode $1013$ and
scale $558$.  This was found by a maximum likelihood
fit, reported in \cite{limb:dero:2010}.  
The Gumbel distribution has an extreme value justification,
but it also gives negative values which are unphysical.
We have replaced it by a Fr\'echet distribution which
is always positive and is also an extreme value 
distribution.  It has
much heavier tails than the Gumbel which one might
need to characterize extreme risks.  Our nominal Fr\'echet
distribution has a similar mean to the Gumbel
distribution above.  The adverse distribution triples
the nominal one.

\begin{figure}[t]
    \centering
    \includegraphics[width=0.9\linewidth]{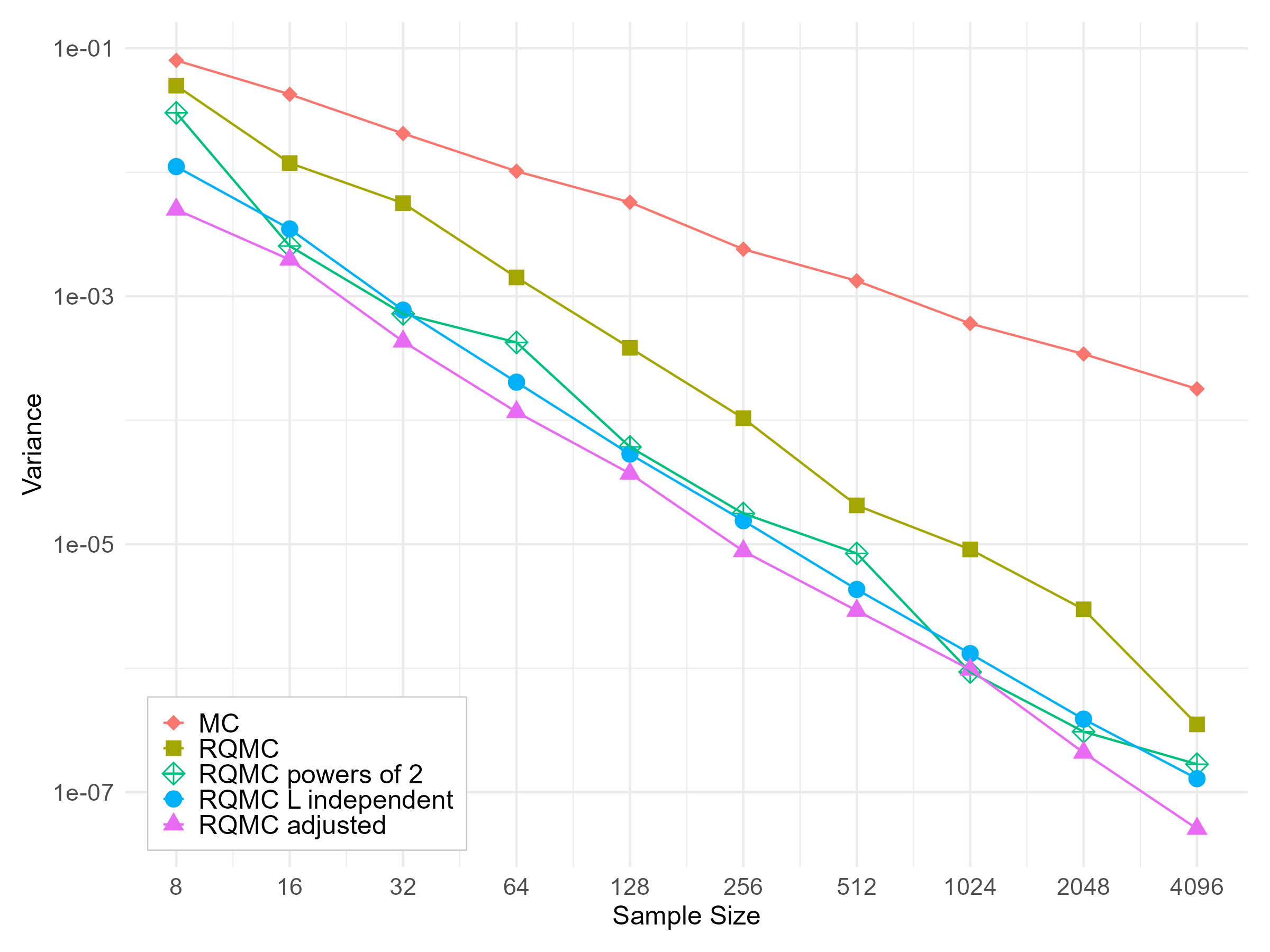}
    \caption{We plot variance versus sample size for $\hat{\mu}_{\mc}$, $\hat{\mu}_{\rqmc}$, $\hat{\mu}_{\rqmc(\adj)}$, $\hat{\mu}_{\rqmc(2)}$ and $\hat{\mu}_{\rqmc(L)}$ for the Saint-Venant flood depth~\eqref{eq:flooddepth}. \changed{The top to bottom ordering of the curves in the legend is the same as they have at $n=4096$.} The sample size allocations for $\hat{\mu}_{\rqmc(\adj)}$, $\hat{\mu}_{\rqmc(2)}$, and $\hat{\mu}_{\rqmc(L)}$ used $\rho=2$.}
    \label{fig:saint-venant-variance}
\end{figure}

The results are shown in Figure~\ref{fig:saint-venant-variance}.
This integrand is unbounded and hence not one
for which the $O(n^{-3+\epsilon})$ results apply,
and so we used $\rho=2$.
We see a large gain from simply replacing MC points
by RQMC points as in $\hat\mu_{\rqmc}$. The RQMC
methods that oversampled the strata with smaller $\alpha_\ell$
had lower variance than the plain RQMC. The best of those
was $\hat\mu_{\rqmc(\adj)}$ which used equation~\eqref{eq:criterion0} 
without requiring $n_\ell$ to be powers of two.  For $\rho=2$,
we don't expect a better rate from taking $n_\ell=2^{m_\ell}$.
Sampling the $L$ strata independently had roughly the
same variance as using the first component of $\bsz_i$
to select a stratum, where we had expected an advantage
for the second method.

In Figure~\ref{fig:enter-label} we consider integer allocations
with different values of $\rho$ using equation~\eqref{eq:criterion0}
to select $n_\ell$.
For large $n$, the lowest variances are for
$\rho=2$ and $3$ with curves that
are nearly parallel with slope close to $-1.8$. The curve for $\rho=2$ 
which is also in Figure~\ref{fig:saint-venant-variance} is best at larger $n$.
The case of $\rho=1$ has proportional allocation and may have
an empty stratum for $n\le 128$ (as $128*0.01<2$).


\begin{figure}[t]
    \centering
    \includegraphics[width=0.9\linewidth]{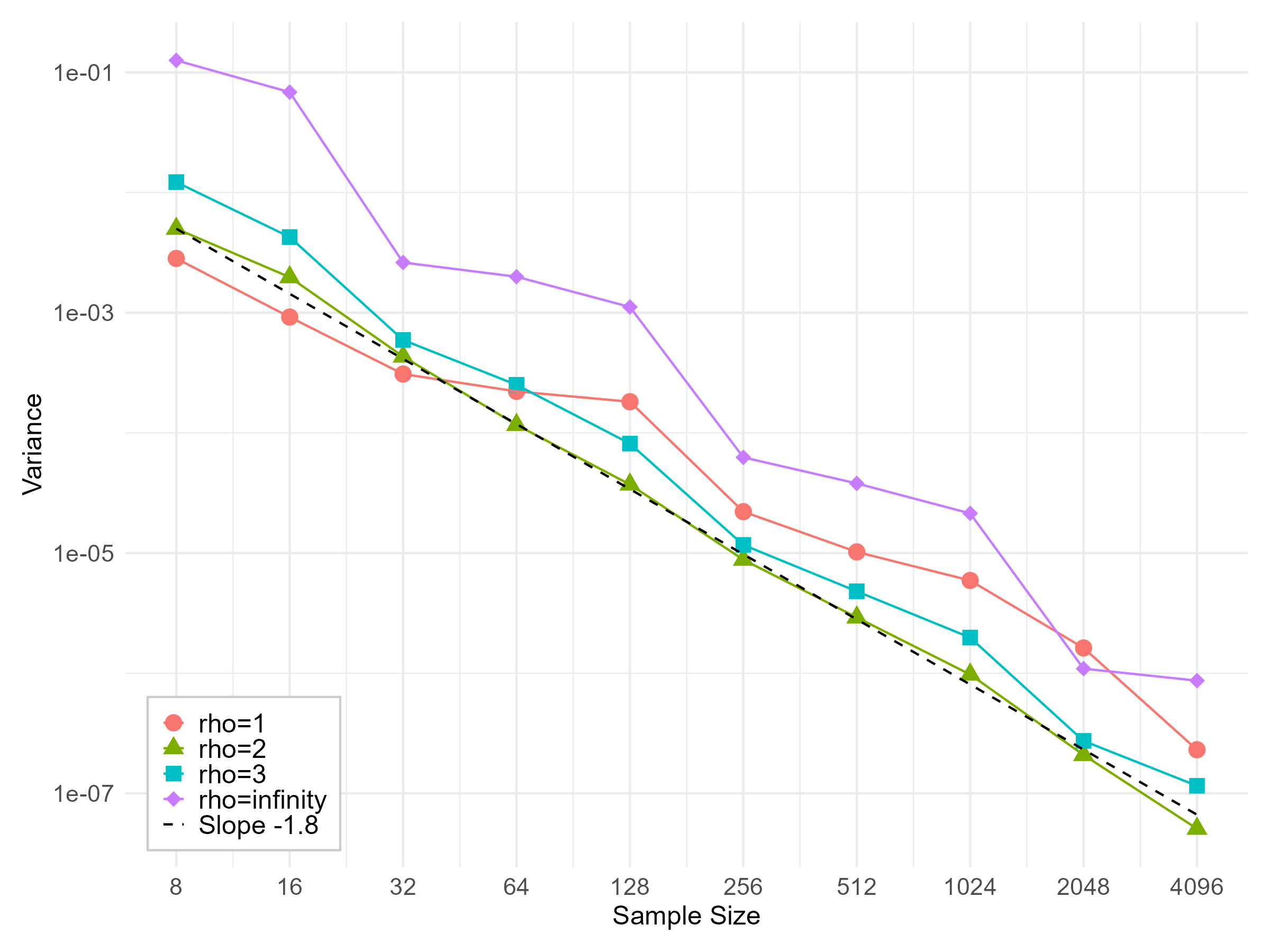}
    \caption{Variance of $\hat{\mu}_{\rqmc(\adj)}$ versus the sample size $n$ for $\rho \in \{1,2,3,+\infty\}$. The reference line has slope $-1.8$.}
    \label{fig:enter-label}
\end{figure}

\section{Discussion}\label{sec:discussion}

We saw RQMC improving the sampling in a
mixture model setting where there is one categorical
variable.  We found that taking account of the better
than Monte Carlo convergence rates motivated changes
where a stratum with mixture probability $\alpha_\ell$
would be better sampled with probability proportional
to $\alpha_\ell^\theta$ for some $\theta<1$.
Choosing the wrong power would adversely affect
the constant in the asymptotic variance but not
the rate.  

A good exponent $\theta$ 
can be derived from an assumed convergence rate, and such a convergence
rate can be estimated adaptively from pilot runs.  We don't
define a formal adaptive algorithm. We are not sure of a
way to do that which would be better than ad hoc experimentation.


We anticipate negative correlations among the $\hat\mu_\ell$
though we can also construct counter-examples.  We expect
that our within-stratum integrands $h_\ell(\bsu)$ will
tend to be similar from stratum to stratum in typical
use cases. As a tiny
example, suppose that $\alpha_1=\alpha_2=1/2$ and
that $h_1(\bsu)=h_2(\bsu)=u_1$. The first $n/2$ values of
$u_1$ (used in $h_1$) combine with the second $n/2$ values
of $u_1$ (used in $h_2$) to give better uniformity than
we get from either half on its own.  We would get
variance $1/(12n^3)$ instead of $2/(12(n/2)^3)=4/(3n^3)$,
which is sixteen times larger.
A similar advantage would be expected when every $h_\ell$ 
is monotone in each $u_j$ either increasing in $u_j$ for all $\ell$
or decreasing in $u_j$ for all $\ell$. For example, increasing
a physical variable might be harmful or helpful
consistently across strata. We can of course
induce positive correlations in the $\hat\mu_\ell$
by taking $h_1(\bsu)=u_1$ and $h_2(\bsu)=-u_1$.
If one expects positive correlations then one might
instead use independently sampled RQMC points
in different strata.  Independent sampling performed
worse than conjoined sampling in the toy problem
but the two were roughly equal in the Saint-Venant
problem. We speculate that the negative correlation advantage,
if it is present, is greatest when several $\beta_\ell$
are large.

\changed{For some problems, the integrand is a simple
 function within each stratum.
 A more interesting expression arises in mixture
 importance sampling.  There we estimate
 $\mu = \int g(\bsx)p(\bsx)\rd\bsx$ for a
 probability density function $p$ by
  \begin{align}\label{eq:mixis}
 \hat\mu = \frac1n \sum_{i=1}^n \frac{g(\bsx_i)p(\bsx_i)}{\sum_{\ell=1}^L\alpha_\ell p_\ell(\bsx_i)}
 \end{align}

 for $\bsx_i\sim\sum_{\ell=1}^L\alpha_\ell p_\ell$.

 We tried a complex importance sampling problem
from \cite{hest:1995} modeling reliability of oil inventory for an energy utility.  
The estimator was of the mixture importance sampling form \eqref{eq:mixis}. An interesting aspect of mixture importance sampling is that it benefits from using $p_{\ell'}(\bsx_i)$ when
$\bsx_i$ was drawn from $p_\ell$ for $\ell\ne\ell'$.
We found only a moderate
improvement with RQMC giving roughly $O(n^{-1.15})$
variance.  The example
was also studied by \cite{li:tan:chen:2013}.
Here we only mention it to show that
RQMC with mixture sampling will not always give an
an extremely large improvement. The relevant
integrands were discontinuous and we speculate
that they may also have had important high order interactions.
  }

We might also have a variable with a mixed distribution
including a continuous component with a density function
and $L$ atoms.  For this setting we can make the continuous
part of that variable's distribution into an $L+1$'st 
mixture component. We will still need to sample within the range
of that continuous component.  For that we can add an $s+1$'st component to $\bsu$, which is then the $s+2$'nd
component of $\bsz$.

We have considered one categorical variable.
Some problems \changed{will} have two or more categorical
variables.  Suppose that there are $K$ categorical variables
where the $k$'th one has $L_k$ levels.
Then we have the choice of whether to use $K$ different
components of RQMC to make the allocations, or to
use one component with $L=\prod_{k=1}^KL_k$ levels
in it.  Optimizing the choice between these two approaches is
outside the scope of this article.

\section*{Acknowledgments}

We thank Neil Sloane for pointing out the connection to sequence A002572.
Thanks to Tim Hesterberg and Wentao Li for help with the
oil inventory problem.  \changed{We thank two anonymous referees and an associate 
editor for comments that helped us improve the paper.}
\bibliography{qmc}

\appendix

\section{Proofs}
This appendix contains the proofs of our propositions.
\subsection{Proof of Proposition~\ref{prop:boundnl}: closeness of $n_*$ to $n\beta$}\label{sec:proof:prop:boundnl}

\begin{proof}    
Write $A=r/n+a$ and $B=s/n+b$ for integers $0\le r\le s<n$ 
and $a,b\in[0,1/n)$. Then $\beta = t/n+c$ for
$t=s-r$ and $-1/n<c<1/n$.
If $r=s$, then $0\le n_*\le1$ and 
then~\eqref{eq:nlbound} holds trivially.

For $s>r$ we may write
$$[A,B)=[A,(r+1)/n)\,\cup\,[(r+1)/n,s/n)\,\cup\,[s/n,s/n+b).$$
Then $[A,B)$ contains at least $s-r-1$
of the intervals $I_k=[(k-1)/n,k/n)$ for $k=1,\dots,n$
and it intersects with at most $s-r+1$ of them.
Therefore $|n_*-(s-r)|\le1$.
Now $\lceil n\beta\rceil=s-r+\lceil nc\rceil$
and so $s-r\ge\lceil n\beta\rceil-1$
establishing the lower limit in~\eqref{eq:nlbound}.
Similarly $\lfloor n\beta\rfloor = s-r+\lfloor nc\rfloor$
so $s-r\le \lfloor n\beta\rfloor+1$
establishing the upper limit in~\eqref{eq:nlbound}.
\end{proof}

\subsection{Proof of Proposition~\ref{prop:stratanets}: digital nets within strata}\label{sec:proof:prop:stratanets}
\begin{proof}
For $S_\ell$ defined at~\eqref{eq:defSell}, sample $i\in S_\ell$ holds if and only if
$B_{\ell-1} \le z_{i1} < B_\ell$ or equivalently
$
\sum_{1\le r<\ell}2^{-\kappa_r}\le z_{i1} < \sum_{1\le r\le\ell}2^{-\kappa_r}
$
which we can write as
\begin{align}\label{eq:instratumell}
 \frac{\gamma}{2^{\kappa_\ell}}\le
z_{i1} < \frac{\gamma+1}{2^{\kappa_\ell}}
\end{align}
for integer $\gamma = \sum_{1\le r<\ell}2^{\kappa_\ell-\kappa_r}\ge0$.
Because the net is stratified, the number of
observations $i$ that satisfy equation~\eqref{eq:instratumell}
is $2^{m-\kappa_\ell}=2^{m_\ell}$.
Then $\{u_i\mid i\in S_{\ell}\}\subset[0,1)^s$  has cardinality $2^{m_\ell}$
making it a $(t_\ell,m_\ell,s)$-net in base $2$ for some value $t_\ell\le m_\ell$.

Now consider a dyadic box in $[0,1]^s$ of the form
$$
\prod_{j=1}^s\Bigl[ \frac{c_j}{2^{k_j}},\frac{c_j+1}{2^{k_j}}\Bigr).
$$
This box has volume $2^{-|\bsk|}$. The number of points $\bsu_i$
from stratum $\ell$ in this box equals the number of points
$\bsz_i\in[0,1)^{s+1}$ that belong to
$$\Bigl[\frac{\gamma}{2^{\kappa_\ell}},\frac{\gamma+1}{2^{\kappa_\ell}}\Bigr)
\times\prod_{j=1}^s\Bigl[
\frac{c_j}{2^{k_j}},\frac{c_j+1}{2^{k_j}}\Bigr).$$
If $m-\kappa_\ell-|\bsk|\ge t$, then
by the digital net property of $\bsz_1,\dots,\bsz_n$, this number is $2^{m-\kappa_\ell-|\bsk|}$.
That is, if $m_\ell\ge|\bsk|+t$, then the number of such points is
$2^{m_\ell-|\bsk|}$.
\end{proof}

\subsection{Proof of Proposition~\ref{prop:stilllowdiscrep}: discrepancy within strata}\label{sec:proof:prop:stilllowdiscrep}
\begin{proof}
The local discrepancy of $\bsz_1,\dots,\bsz_n$ 
at the point
$(b,\bsc)\in[0,1]^{s+1}$ is
$$
\delta( b,\bsc) = \frac1n\sum_{i=1}^n\ind\{0\le v_i<b\}\ind\{\bsz_i\in[\bszero,\bsc)\} - b\times\vol([\bszero,\bsc))
$$
where $\vol([\bszero,\bsc))=\prod_{j=1}^sc_j$. 
Because $\bsz_1,\dots,\bsz_n$ are a digital
net, $|\delta(b,\bsc)|=O( n^{-1}\log(n)^{s})$.

The sample points for stratum $\ell$ are
the $n_\ell$ points $\bsx_i$ for which $B_{\ell-1}\le v_i<B_\ell$.
By equation~\eqref{eq:nlbound}, $n_\ell = n\beta_\ell+O(1)$.
The local discrepancy within stratum $\ell$ at $\bsc\in[0,1)^s$ is
\begin{align*}
&\ \quad\frac1{n_\ell}\sum_{i=1}^n\ind\{B_{\ell-1}\le v_i<B_\ell\}
\ind\{\bsx_i\in[\bszero,\bsc)\}-\vol([\bszero,\bsc))\\
&=\frac{n}{n_\ell}\Bigl(
\beta_\ell\vol([\bszero,\bsc))+\delta(B_\ell,\bsc)-\delta(B_{\ell-1},\bsc)
\Bigr)
-\vol([\bszero,\bsc)).
\end{align*}
Its absolute value is at most
\begin{align*}
\Bigl|\frac{n\beta_\ell}{n_{\ell}}-1\Bigr|\vol([\bszero,\bsc))
+\frac{2n}{n_\ell}O(n^{-1}\log(n)^{s})
&=O(n_{\ell}^{-1}\log(n_\ell/\beta_\ell)^{s})\\
&=O( n_\ell^{-1} \log(n_\ell)^{s})
\end{align*}
and so the star discrepancy of $\bsu_i$ for $i\in S_\ell$
is $O( n_\ell^{-1} \log(n_\ell)^{s})$.
\end{proof}

\subsection{Proof of Proposition~\ref{prop:monotone}: monotonocity of $I_0$}\label{sec:proof:prop:monotone}
\begin{proof}
For fixed $\rho$ we may write
$$
I_0(\gamma\giv\rho)\propto \tilde I_0(\gamma\giv\rho)
= \sum_{\ell=1}^L\alpha_\ell^{2-2\rho/(\gamma+1)}\Biggl(\sum_{k=1}^L\alpha_k^{2/(\gamma+1)}\Biggr)^\rho,
$$
so it is enough to show that~\eqref{eq:monosubopt} holds with
$I_0$ replaced by $\tilde I_0$.
Now write $2/(\gamma+1) = 2/(\rho+1)+x$, so we can use $x$
to characterize how $\gamma$ differs from $\rho$
with $x=0$ corresponding to $\gamma=\rho$.
Then we can write $\tilde I_0(\gamma\giv\rho)=h(x)$ where
\begin{align*}
h(x) &=\biggl( \,\sum_{\ell=1}^L\alpha_\ell^{2-2\rho/(\rho+1)-\rho x}\biggr)
\biggl(\,\sum_{k=1}^L\alpha_k^{2/(\rho+1)+x}\biggr)^\rho
 =\biggl( \,\sum_{\ell=1}^L\alpha_\ell^{2/(\rho+1)-\rho x}\biggr)
\biggl(\,\sum_{k=1}^L\alpha_k^{2/(\rho+1)+x}\biggr)^\rho.
\end{align*}
It is enough to show that $h(x)$ is non-decreasing as $x$ either increases
from zero or decreases from zero.  We will do this by showing
that $xh'(x)\ge0$.

For integers $r\ge0$ define 
\begin{align*}
S_+(r) &= S_+(r,x)=\sum_{\ell=1}^L\alpha_\ell^{2/(\rho+1)+x}\log(\alpha_\ell)^r
\quad\text{and}\\
S_-(r) &= S_-(r,x)=\sum_{\ell=1}^L\alpha_\ell^{2/(\rho+1)-\rho x}\log(\alpha_\ell)^r.
\end{align*}
Their derivatives with respect to $x$ are
$S_+(r+1)$ and $-\rho S_-(r+1)$, respectively.
Also, $S_+(r,0)=S_-(r,0)$ which we write as $S(r)$.

Now we write
$h(x) = S_-(0)S_+(0)^\rho$
so
$h'(x) = -\rho S_-(1)S_+(0)^\rho+\rho S_-(0)S_+(0)^{\rho-1}S_+(1)$.
Then
\begin{align*}
\frac{h'(x)}{\rho S_+(0)^{\rho-1}} &=
S_+(1)S_-(0)-S_+(0)S_-(1)\\
&=\sum_{i=1}^L\sum_{j=1}^L
\alpha_i^{2/(\rho+1)-\rho x}\alpha_j^{2/(\rho+1)+x}\log(\alpha_j)
-\alpha_i^{2/(\rho+1)-\rho x}\alpha_j^{2/(\rho+1)+x}\log(\alpha_i)\\
&=\sum_{i=1}^{L-1}\sum_{j=i+1}^L \log\Bigl(\frac{\alpha_j}{\alpha_i}\Bigr)c_{ij}
\end{align*}
where
\begin{align*}
c_{ij}&=
\alpha_i^{2/(\rho+1)-\rho x}\alpha_j^{2/(\rho+1)+ x}
-\alpha_i^{2/(\rho+1)+ x}\alpha_j^{2/(\rho+1)-\rho x}\\
&=\alpha_i^{2/(\rho+1)} \alpha_j^{2/(\rho+1)} \Bigl(\Bigl(\frac{\alpha_j}{\alpha_i^\rho}\Bigr)^x-\Bigl(\frac{\alpha_i}{\alpha_j^\rho}\Bigr)^x\Bigr).
\end{align*}

If $\alpha_j\geq \alpha_i$, then $\alpha_j/\alpha_i^\rho\geq\alpha_i/\alpha_j^\rho$ and so $x c_{ij}\geq 0$. Because $\log(\alpha_j/\alpha_i)\geq 0$, we conclude that $x\log(\alpha_j/\alpha_i)c_{ij}\geq 0$. If instead $\alpha_j\le \alpha_i$, then a similar argument shows that $xc_{ij}\le0$ and once again $x\log(\alpha_j/\alpha_i)c_{ij}\geq 0$.
Hence
$$xh'(x)=\rho S_+(0)^{\rho-1}\sum_{i=1}^L\sum_{j=1}^L x \log(\alpha_j/\alpha_i)c_{ij} \geq 0$$
as required.
\end{proof}

\subsection{Proof of Proposition \ref{prop:beton3vs1whenits2}: $I_0(3\giv 2)\le I_0(1\giv 2)$ }
\label{sec:proof:prop:beton3vs1whenits2}
\begin{proof}
From equation~\eqref{eq:i0ineff}
$$
I_0(1\giv2) = \frac{L}{\bigl(\sum_{\ell=1}^L\alpha_\ell^{2/3}\bigr)^3}
\quad\text{and}\quad
I_0(3\giv2) = \frac{\bigl(\sum_{\ell=1}^L\alpha_\ell^{1/2}\bigr)^2}{\bigl(\sum_{\ell=1}^L\alpha_\ell^{2/3}\bigr)^3}.
$$
These expressions have the same denominator and then the
result follows from the Cauchy-Schwarz inequality applied
to the numerators.
\end{proof}

\subsection{Proof of Proposition~\ref{prop:alwaysterminates}: sure termination}
\label{sec:proof:prop:alwaysterminates}
\begin{proof}
First, if $n'>0$ then we can show that there is an `eligible' index
$\ell$ with $n_\ell\le n'$, that can be doubled
without causing $\sum_{\ell=1}^Ln_\ell>n$. At any stage in the
algorithm, $n_{\ell^*}$ for $\ell^*=\argmin_{1\le \ell\le L}n_\ell$ is a power of two.  It divides $n$ and it divides every
$n_\ell$, so it also divides the nonzero value $n'=n-\sum_{\ell=1}^Ln_\ell$, hence $n_{\ell^*}\le n'$. Finally, every step with $n'>0$ reduces $n'$, and so $n'$ must reach zero in a finite
number of steps starting from $n'=n-L$.
\end{proof}

\subsection{Proof of Proposition~\ref{prop:minmax}: minimax sample sizes}
\label{sec:proof:prop:minmax}
\begin{proof}
The same argument applies to both criteria.  We prove it
for criterion 1. Fixing $\bsn$ and $\tilde\bsn$ we
see that the maximum over $\bstau$ of $R_1(\bsn\giv\tilde\bsn;\bstau)$ arises
when $\tau_\ell=0$ for all $\ell$ except 
$\ell^*(\bsn,\tilde\bsn)=\argmax_\ell (\tilde n_\ell/n_\ell)^{\rho/2}$ with $\tau_{\ell^*(\bsn,\tilde\bsn)}$ then taking any value $\lambda>0$.

Now, we can go back to the initial problem and rewrite \eqref{eq:mainpb} as:
$$\min_{\bm{n}\in \Delta_N}\max_{\tilde{\bm{n}}\in \Delta_N}
\Bigl(\frac{\tilde{n}_{\ell^*(\bsn,\tilde\bsn)}}{{n}_{\ell^*(\bsn,\tilde\bsn)}}\Bigr)^{\rho/2}
$$
after noting that the factor $\alpha_{\ell^*(\bsn,\tilde\bsn)}\tau_{\ell^*(\bsn,\tilde\bsn)}^{1/2}$ cancels
from numerator and denominator.
To maximize this ratio over $\tilde\bsn$ one takes
$\tilde\bsn$ to be the vector with value $N-L+1$
in component $\ell^*=\argmin_\ell n_\ell$ and ones elsewhere. Then the minimax
problem becomes
$\min_{\bsn\in\Delta_N}\bigl((N-L+1)/\min_{1\le\ell\le N} n_\ell\bigr)^{\rho/2}$.  The solution is to maximize the
minimum value of $n_\ell$.
We do that by making $\min_\ell n_\ell=\lfloor N/L\rfloor$.
\end{proof}

\subsection{\changed{Varying $\rho_\ell$}}\label{sec:varyrho}
\changed{
Our model for sample allocation is that we start with $n_\ell=1$ for $\ell=1,\dots,L$.
Then we repeatedly increment $n_{\ell^*}$ by one for
$
\ell^*=\arg\max_\ell\alpha_\ell^2(n_\ell^{-\rho_\ell}- (n_\ell+1)^{-\rho_\ell})
$
because this gives the stratum that most reduces the criterion.
For our RQMC estimate to be consistent, we require $\min_\ell n_\ell\to\infty$, and so we assume
that all of the $n_\ell$ are large enough to allow a Taylor approximation below.
We also ignore ties in this analysis.

Using $(n_\ell+1)^{-\rho_\ell}=n_\ell^{-\rho_\ell}-\rho_\ell n_\ell^{-(\rho_\ell+1)} +O(n_\ell^{-\rho_\ell-2})$,
we simplify the selection rule to
$
\ell^*=\arg\max_\ell\alpha_\ell^2\rho_\ell/n_\ell^{\rho_{\ell}+1}.
$
Then for $1\le \rho_\ell <\rho_{\lp}\le 3$ under our greedy allocation
$n_{\lp}$ will only get a new point if
$$
n_{\lp} <
\Bigl(\frac{\alpha_{\lp}^2\rho_{\lp}}{\alpha_\ell^2\rho_\ell}n_\ell^{\rho_\ell+1}
\Bigr)^{1/(\rho_{\lp}+1)}
$$
and so we expect $n_{\lp} \lesssim n_\ell^{(\rho_\ell+1)/(\rho_{\lp}+1)}$, where $\lesssim$ denotes inequality up to a multiplicative constant.
For the most extreme case, with $\rho_\ell=1$ and $\rho_{\lp}=3$, we get
 $n_{\lp} \lesssim n_\ell^{1/2}$.
 }




\end{document}